\documentclass[10pt,conference,letterpaper]{IEEEtran}
\IEEEoverridecommandlockouts

\makeatletter
\long\def\@makecaption#1#2{\ifx\@captype\@IEEEtablestring%
\footnotesize\begin{center}{\normalfont\footnotesize #1}\\
{\normalfont\footnotesize\scshape #2}\end{center}%
\@IEEEtablecaptionsepspace
\else
\@IEEEfigurecaptionsepspace
\setbox\@tempboxa\hbox{\normalfont\footnotesize {#1.}~~ #2}%
\ifdim \wd\@tempboxa >\hsize%
\setbox\@tempboxa\hbox{\normalfont\footnotesize {#1.}~~ }%
\parbox[t]{\hsize}{\normalfont\footnotesize \noindent\unhbox\@tempboxa#2}%
\else
\hbox to\hsize{\normalfont\footnotesize\hfil\box\@tempboxa\hfil}\fi\fi}
\makeatother

\usepackage{amsmath}
\usepackage{amssymb}
\usepackage{overpic}
\usepackage{multirow}
\usepackage{subfigure}
\usepackage[ruled, linesnumbered, noend]{algorithm2e}
\usepackage{booktabs}
\usepackage{float}
\usepackage{epsfig}

\newcommand{\set}[1]{\{ #1 \}}
\newcommand{\setsep}{\;|\;}
\newcommand{\ov}[1]{\, \overline{\! \raisebox{0pt}[\dimexpr\height+0.2mm\relax]{$\displaystyle #1 $} \!\!}\,\,}
\newcommand{\un}[1]{\, \underline{\! #1 \!}\,}
\newcommand{\tuple}[1]{\langle #1 \rangle}
\newcommand{\shrink}{\vspace{-10pt}}
\newcommand{\eat}[1]{}

\newtheorem{definition}{\hspace{-12pt} Definition}
\newtheorem{proposition}{\hspace{-12pt} Proposition}
\newtheorem{theorem}{\hspace{-12pt} Theorem}
\newtheorem{lemma}{\hspace{-12pt} Lemma}
\newtheorem{corollary}{\hspace{-12pt} Corollary}

\begin{document}

\title{Reverse Nearest Neighbor Heat Maps: A Tool for Influence Exploration\vspace{-.7\baselineskip}}

\author{
\IEEEauthorblockN{
Yu Sun\IEEEauthorrefmark{1},
Rui Zhang{\small $^\S$}\IEEEauthorrefmark{1},
Andy Yuan Xue\IEEEauthorrefmark{1},
Jianzhong Qi\IEEEauthorrefmark{1},
Xiaoyong Du\IEEEauthorrefmark{2}
}
\IEEEauthorblockA{\IEEEauthorrefmark{1}Department of Computing and Information Systems, University of Melbourne\\\IEEEauthorrefmark{2}Renmin University of China and Key Laboratory of Data Engineering and Knowledge Engineering, MOE, China}
\IEEEauthorblockA{Email: \IEEEauthorrefmark{1}\{sun.y, rui.zhang, andy.xue, jianzhong.qi\}@unimelb.edu.au\hspace{0.5em}\IEEEauthorrefmark{2}duyong@ruc.edu.cn}
\vspace{.6\baselineskip}
\thanks{\vspace{-5pt}\footnotesize{$^\S$}Corresponding author. \quad \quad \quad \quad \quad \quad Copyright~\copyright~2016 IEEE}
}

\maketitle
\begin{abstract}
We study the problem of constructing a reverse nearest neighbor (RNN) heat map by finding the RNN set of every point in a two-dimensional space. Based on the RNN set of a point, we obtain a quantitative \emph{influence} (i.e., \emph{heat}) for the point. The heat map provides a global view on the influence distribution in the space, and hence supports exploratory analyses in many applications such as marketing and resource management. To construct such a heat map, we first reduce it to a problem called \emph{Region Coloring} (RC), which divides the space into disjoint \emph{regions} within which all the points have the same RNN set. We then propose a novel algorithm named CREST that efficiently solves the RC problem by labeling each region with the heat value of its containing points. In CREST, we propose innovative techniques to avoid processing expensive RNN queries and greatly reduce the number of region labeling operations. We perform detailed analyses on the complexity of CREST and lower bounds of the RC problem, and prove that CREST is asymptotically optimal in the worst case. Extensive experiments with both real and synthetic data sets demonstrate that CREST outperforms alternative algorithms by several orders of magnitude.
\end{abstract}

\section{Introduction} \label{sec:intro}
In market analysis, urban design, and facility placement, we often need to select a suitable location for new facilities such as a warehouse or a hospital.
Emerging event-based social networks such as Meetup and Whova also need to select an appropriate location suitable for the event-participant arrangement.
These problems are called the \emph{location selection} problem, which is usually a multi-criteria decision making process involving various quantitative and qualitative factors.
A quantitative factor usually considered is the \emph{influence} of the location, which is commonly measured by the \emph{reverse nearest neighbor} (RNN) set of the location~\cite{Korn2000,Sun2012cikm,Wong2011}.
Given two sets of points $\mathcal{O}$ and $\mathcal{F}$, the RNN set of a location $p$ is a subset of $\mathcal{O}$ that are closest to $p$ among all the points in $\mathcal{F}$.
There are many ways to measure the influence of $p$ by the RNN set. 
Straightforward measures consider only the size or total weight of the set~\cite{Cabello2010,Wong2011,Zhou2011}.
Other measures consider various attributes of the data points in $\mathcal{O}$ and $\mathcal{F}$, such as the capacity constraint~\cite{Melkote2001,Sun2012cikm}, social relationship~\cite{She:2015:USE,Yang:2012:SGQ}, etc.
While we can model the quantitative factors precisely by numbers, we can not do the same to many qualitative factors such as the area safety, demographic composition and convenience of public transportation.
Some factors in decision-making are also vague and imprecise, which are subject to decision maker's judgments.
To assist decision making based on quantitative measures while still allowing subjective judgments based on qualitative measures and other factors, we introduce the RNN heat map, which shows the influence (quantitative measures) of every point in the space.
Comparing to existing studies~\cite{chen2014efficient,gao2015efficient,Sun2012cikm,Wong2011,Zhou2011} which give only the points or regions with the highest influence, the RNN heat map enables exploring the influence of the whole space while considering qualitative factors at any instant during the exploration.
\begin{figure}
\centering
\subfigure[RNN heat map]{
\begin{overpic}[scale=.28]{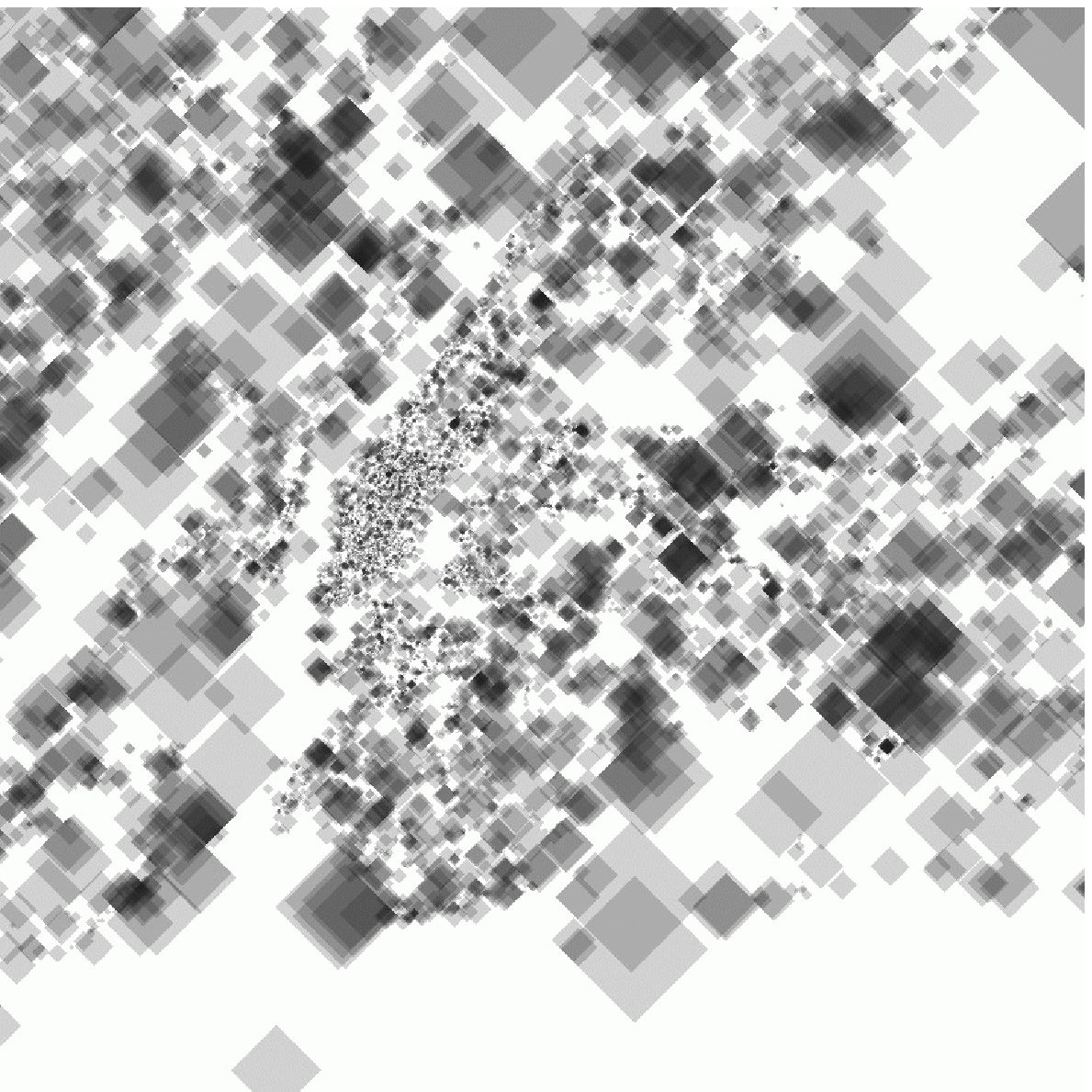}
\end{overpic}
\label{fig:ny_hm}
}
\subfigure[Satellite map]{
\begin{overpic}[scale=.28]{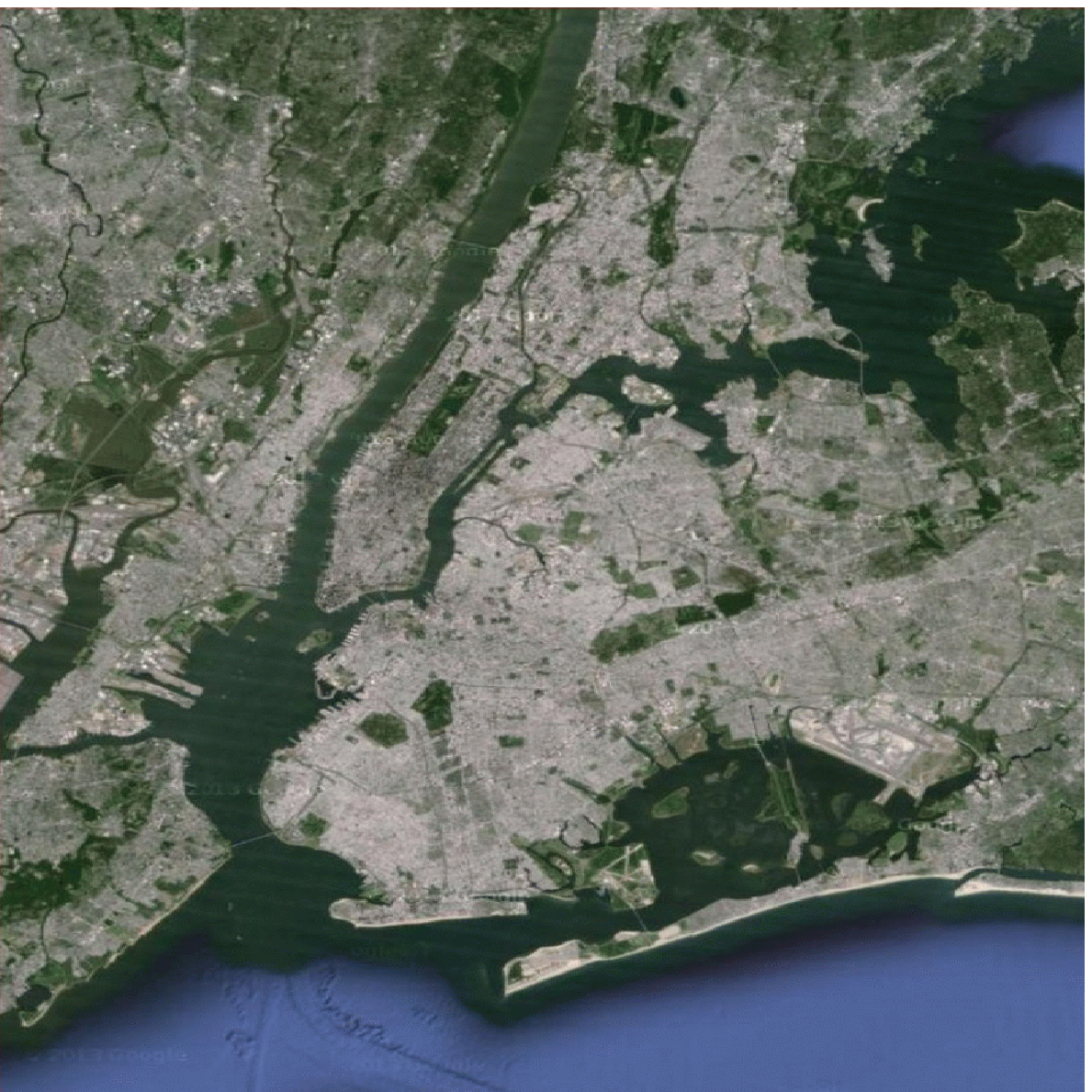}
\label{fig:ny_gm}
\end{overpic}
}
\vspace{-8pt}
\caption{RNN Heat Map of New York City}
\shrink
\vspace{-10pt}
\end{figure}
Consider a scenario where RNN heat maps are used to assist selecting locations of self-pickup and drop-off service points for courier companies.
Let $\mathcal{O}$ be the potential clients and $\mathcal{F}$ be the existing service points.
For simplicity, let the size of the RNN set measure the influence, i.e., \emph{heat} (although any other functions related to the RNN set can be used).
Fig.~\ref{fig:ny_hm} shows such an RNN heat map for the New York City, whose satellite image is shown in Fig.~\ref{fig:ny_gm}.
The darker regions indicate higher heat values.
Such a heat map will allow the exploration of influential regions while considering qualitative factors as discussed above.
Note that regions with high influence values do not necessarily correspond to regions of high client density because we need to consider the competition from existing facilities.
For example in Fig.~\ref{fig:density}, the upper left corner has the highest client density, but the most influential and the $4^{th}$ influential regions are in the middle, denoted by the two gray rectangles (the $2^{nd}$ and the $3^{rd}$ most influential regions are also in the middle near these two but too small to be visible).
Without the RNN heat map, it is very difficult or impossible to explore all these different choices and make well-informed decisions.
\begin{figure*}[!t]
\begin{minipage}[t]{0.27\linewidth} 
\centering
\begin{figure}[H] 
\centering
\begin{overpic}[scale=.61]{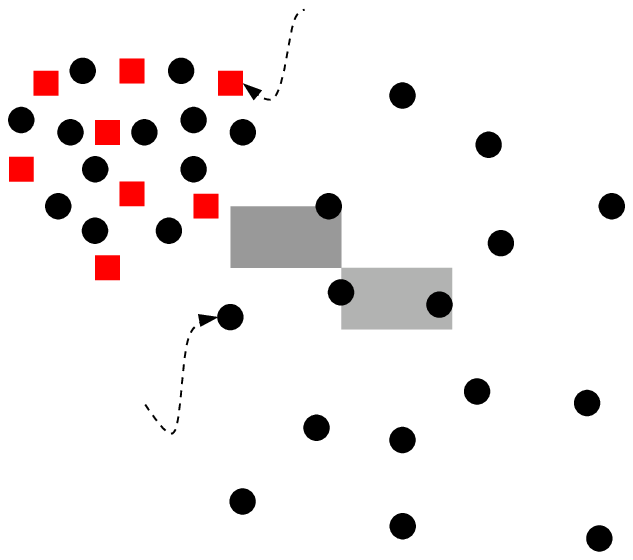}
\small{
\put(7,25){clients} \put(51,81){existing facilities}
}
\end{overpic}
\shrink
\vspace{-5pt}
\caption{Client density}
\shrink
\label{fig:density}
\end{figure}
\vspace{-5pt}
\end{minipage}
\hspace{0.04\linewidth}
\begin{minipage}[t]{0.72\linewidth}
\centering
\begin{figure}[H] 
\centering
\subfigure[]{
\begin{overpic}[scale=.52]{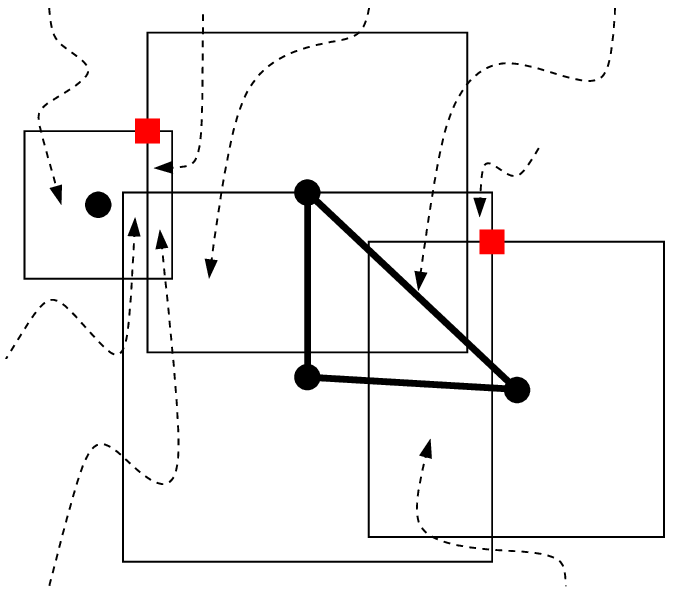}
\small{
\put(8,49){$o_3$} \put(37,62){$o_4$} \put(40,25){$o_1$} \put(78,22){$o_2$} \put(14,70){$f_1$} \put(74,53){$f_2$}
\put(12,86){$\set{o_3,o_4}$} \put(-5,86){$\set{o_3}$} \put(-13,28){$\set{o_1,o_3}$} \put(-10,-5){$\set{o_1,o_3,o_4}$} 
\put(30,12){$\set{o_1}$} \put(41,86){$\set{o_1,o_4}$} \put(40,70){$\set{o_4}$}
\put(72,66){$\set{o_1}$} \put(70,86){$\set{o_1,o_2,o_4}$} \put(70,-5){$\set{o_1,o_2}$} \put(78,37){$\set{o_2}$}
}
\end{overpic}\label{fig:intro}
}
\subfigure[A superimposition]{
\begin{overpic}[scale=.52]{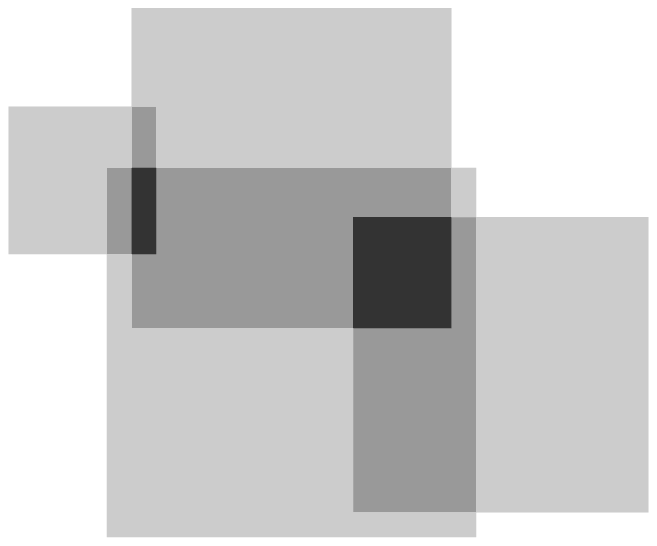}
\end{overpic}\label{fig:superim}
}
\subfigure[A heat map]{
\begin{overpic}[scale=.52]{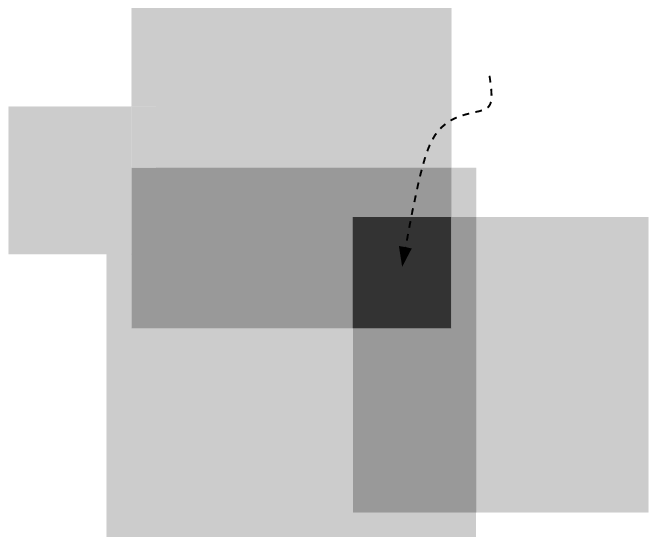}
\small{
\put(35,43){$1.0$} \put(57,20){$1.0$}
\put(69,72){$3.0$}
}
\end{overpic}\label{fig:heatmap}
}
\vspace{-7pt}
\caption{An example of the RC problem}
\shrink
\end{figure}
\vspace{-5pt}
\end{minipage}
\vspace{-1.5em}
\end{figure*}

To construct such a heat map, we need to obtain the influence value of every point in the space.
We call such a problem the \emph{RNN heat map} (RNNHM) problem:
\begin{definition}[RNN Heat Map Problem]
Given two sets of points $\mathcal{O}$ and $\mathcal{F}$ and a distance metric in a two-dimensional space, the RNN set of a point $q$ ($q \notin \mathcal{F}$) is a subset of $\mathcal{O}$ that have $q$ as their nearest neighbor comparing with other points in $\mathcal{F}$.
Given any influence measure, which is a real-valued function on the RNN set, associate each point in the space with its influence value, i.e., the heat value.
\end{definition}
Since the number of points in the space is infinite, to solve the RNNHM problem, we first reduce it to a problem called \emph{Region Coloring} (RC), which divides the space into disjoint \emph{regions}, within which all the points have the same RNN set (detailed in Section~\ref{sec:proFor}).
We use Fig.~\ref{fig:intro} to illustrate the RC problem with $L_{\infty}$.
For simplicity and ease of presentation, we will first discuss how to obtain such regions and compute their influence values with the $L_{\infty}$ metric, and then extend the techniques to $L_{1}$ and $L_{2}$ metrics.
In Fig.~\ref{fig:intro}, let $\mathcal{O} = \set{o_1, o_2, o_3, o_4}$, represented by the black dots, and $\mathcal{F} = \set{f_1,f_2}$, represented by the small red squares.
For each point $o$ in $\mathcal{O}$, we draw a ``circle" called the \emph{NN-circle} with $o$ being the center and the distance to $o$'s nearest neighbor (NN) being the radius, which is a square with the $L_{\infty}$ metric.
The NN-circles partition the space into separate regions.
It can be proved that all the points in such a region have the same RNN set.
The RC problem is to obtain the influence of each such region in the space.

Note that if we measure the influence simply by the size of the RNN set, we can build the heat map by a \emph{superimposition} of the NN-circles, i.e., overlap/overlay of translucent NN-circles as shown in Fig.~\ref{fig:superim}. A darker region suggests more NN-circles overlapping there and hence a higher influence.
However, for a more generic influence measure than the size (or a weighted sum of the RNNs), the heat map can not be achieved by such a simple superimposition.
For example, consider a taxi-sharing scenario~\cite{Ma:2014:real} where the heat map assists taxi drivers to decide the next pick-up locations.
In Fig.~\ref{fig:intro}, let $\mathcal{O}$ be potential passengers, e.g., users of taxi booking apps, and $\mathcal{F}$ be taxis.
Assume that taxi drivers make more profits when taking together multiple passengers whose destinations are close, say within one kilometer.
Let the data points $o_1$, $o_2$, and $o_4$ connected by an edge denote such passengers.
Under such setting, the influence of a location becomes the number of connected passengers in the RNN set.
We build the heat map as shown in Fig.~\ref{fig:heatmap}. We can see that there is only one darkest region, which has an influence value of $3.0$ since its RNN set is $\set{o_1,o_2,o_4}$ and there are three edges connecting $o_1$, $o_2$ and $o_4$.
In comparison, the superimposition as shown in Fig.~\ref{fig:superim} creates two darkest regions, both have an influence value of $3.0$, one with the RNN set $\set{o_1,o_3,o_4}$ and the other $\set{o_1,o_2,o_4}$.
Under the measure that favors connected data points, the RNN set $\set{o_1,o_3,o_4}$ only has an influence value of $1.0$, which is not a good choice for picking up passengers.
Another example is that in the previous courier company scenario, all the service points have a capacity limit (e.g., the storage space). 
Taking these attributes into account, the influence of a location will depend not only on the size of the RNN set but also on its serving capacity\footnote{\scriptsize The influence of a location $p$ is computed by $\sum_{f \in \mathcal{F}\cup{\{p\}}}{\min \{c(f), |\mathcal{R}(f)|\}}$, where $c(f)$ is the capacity and $\mathcal{R}(f)$ the RNN set of $f$~\cite{Sun2012cikm}.\label{foot:func}}.
The superimposition will not be able to handle such influence measures.

Besides not being able to compute the RNN heat map for generic influence measures, a superimposition also cannot support interactive post-processing operations such as selectively showing regions with heat values above a threshold or regions having the top-$k$ heat values, whereas these operations can be easily applied as post-processing of our proposed techniques, which aim to obtain the RNN set of every region in the space.

In this paper, we investigate algorithms to efficiently solve the RNN heat map problem.
In some applications such as taxi-sharing, the heat map may change as clients move around and need to be recomputed frequently. Therefore, an efficient algorithm to the RNNHM problem is crucial.
A straightforward approach such as employing a grid to divide the space and then using the cells to fit the regions has difficulties in finding the right granularity and suffers from low efficiency.
When the influence measure involves a large amount of attributes such as the capacities of taxis and connections of clients, it can also be very expensive to compute~\cite{Sun2012cikm}.
To overcome these challenges, we propose an innovative algorithm named CREST (\textbf{C}onstructing \textbf{R}NN h\textbf{E}at map with the \textbf{S}weep line s\textbf{T}rategy) which efficiently solves the RNNHM problem.
Through a detailed analysis, we prove that CREST is asymptotically optimal in the worst case.
CREST is also generic in the sense that it applies to any influence measure computable from RNN sets and can easily support interactive post-processing operations as described above.
The main contributions of this paper are summarized as follows.
\begin{itemize}
\setlength\itemsep{4pt}
\item We propose the RNN heat map problem, which computes a heat map showing the distribution of RNN-based scores to support effective exploratory analyses.
%
%
\item We propose an innovative algorithm named CREST which efficiently solves the RNN heat map problem. The algorithm utilizes two novel techniques to respectively avoid processing any RNN queries and greatly reduce the times of influence computation.
\item We carefully analyze the complexity of CREST and
lower bounds of the RC problem, and prove that CREST is asymptotically optimal in the worst case.
%
\item We also conduct extensive experiments with both real and synthetic data sets. The results confirm the superiority of CREST by showing that CREST outperforms alternative algorithms by several orders of magnitude.
%
\end{itemize}
The remainder of this paper is organized as follows. 
Section~\ref{sec:relatedwork} reviews related work. Section~\ref{sec:proFor} formalizes the problem. Section~\ref{sec:enlighten} discusses a baseline algorithm. Section~\ref{sec:sweepalgo} describes the CREST algorithm. Section~\ref{sec:analysis} analyzes the complexity. Section~\ref{sec:extension} extends CREST to other settings.
Section~\ref{sec:exp} shows the experiments and Section~\ref{sec:conclu} concludes the paper.

\section{Related Work} \label{sec:relatedwork}
\textbf{RNN Query}. The RNN query is introduced by Korn et al.~\cite{Korn2000}.
Yang et al.~\cite{Yang2001} proposed the Rdnn-tree (a variant of R-tree) to process the RNN query.
Maheshwari et al.~\cite{Maheshwari2002} present a data structure for answering the monochromatic RNN query by utilizing a persistent search tree~\cite{Sarnak:1986:PPL}.
The structure first obtains the NN-circles enclosing a query point in the $x$ dimension and then among these retrieved NN-circles locates the face (region) enclosing the point in the $y$ dimension.
These algorithms focus on computing the RNN set of a single query point. None of them directly applies to the RNNHM problem.
In the RNNHM problem, the aim is to compute the RNN set for every point in the space all at the same time, and the challenge is to avoid the expensive RNN computation.
The All Nearest Neighbor (ANN)~\cite{chen2007efficient} operation takes as input two discrete and finite sets of points and computes for each point in the first set the NN in the second set. For the RNNHM problem, however, we need to obtain RNN sets for essentially infinite points in a continuous space. Therefore, the techniques for ANN do not apply.

\textbf{Influence Measures based on RNN Sets}.
Various influence measures based on the RNN set have been studied.
Korn et al.~\cite{Korn2000} propose to use the size (or sum of weights) of RNN sets as the influence value.
To find the optimal points whose RNN sets are of the maximum size (influence), Cabello et al.~\cite{Cabello2010} propose the maximization problem MaxCov and they solve the problem by finding the depth of an arrangement of disks.
Wong et al.~\cite{Wong2011} solve MaxCov by the devised MaxOverlap algorithm.
Huang et al.~\cite{Huang2011} and Xia et al.~\cite{Xia2005} investigate finding such points in a given set.
Sun et al.~\cite{Sun:2012:topk,Sun2012cikm,Sun:2015:MBL} additionally consider the capacity constraints of such points and study how to achieve a global influence maximization instead of a local maximization. 
Qi et al.~\cite{qi2012min} define the influence based on the average distance between a point and its RNNs.
As RNNHM applies to a general measure, the RNNHM problem can be viewed as a generalized version of the above problems and therefore the solution of RNNHM can be adapted to solve these problems. However, their solutions do not apply to RNNHM, since the special properties exploited in these problems do not present in RNNHM.

\textbf{RNN Variations}.
RNNHM is a variant of the RNN query.
There are also many other studies on variations of the RNN query.
For instance,
Lu et al.~\cite{Lu2011} investigate reverse spatial and textual nearest neighbor queries, in which both location and textual descriptions are considered in the distance metric.
Similarly, Sun et al.~\cite{Sun:2015:knnta} consider temporal aggregates of location-based social network check-ins in the distance metric.
Zhang et al.~\cite{zhang2010hv} design indexes utilizing modern memory hierarchies to speed up such query processing. Ali et al.~\cite{ali2008motion} study approaches to continuous retrieval of the query objects.
She et al.~\cite{She:2015:USE,she2015conflict} devise algorithms to arrange social events to proper users using RNN sets.
These problems are quite different from RNNHM and the proposed algorithms cannot be adapted to solve the RNNHM problem.

\textbf{Sweep Line Strategy}.
%
%
The sweep line strategy is a quite generic approach to handling geometric objects. The Bentley--Ottmann (BO) algorithm uses this strategy to compute intersections of line segments~\cite{Bentley1979} or rectangles~\cite{Bentley1980}.
The BO algorithm and the proposed CREST algorithm compute very different problems and are different in many aspects. i) BO computes only pairwise intersections of line segments or rectangles, while CREST computes the overlaps and relative complements of multiple circles, squares, and axis-aligned line segments, which are much more challenging. ii) In order to efficiently compute the RNN sets, besides the line status, CREST need to memorize the RNN sets of previous events. This requires a delicate design to minimize the overhead and achieve optimal performance. BO does not have such optimization.
%

\section{Problem Formulation} \label{sec:proFor}

We first introduce basic concepts in Section~\ref{subsec:prelimi} and then reduce RNNHM to the Region Coloring (RC) problem in Section~\ref{subsec:equiPro}. Frequently used symbols are listed in Table~\ref{tab:symbol}.
\begin{table}[htb]
\vspace{-5pt}
\small
\centering
\renewcommand\arraystretch{1.2}
\caption{Frequently Used Symbols}\label{tab:symbol}
\shrink
\begin{tabular}{l|l}
\toprule
\textbf{Symbol} & \textbf{Meaning} \\
\midrule \midrule
$\mathcal{O}$ 		&	the set of clients	\\ \hline
$\mathcal{F}$		&	the set of facilities	\\ \hline
$n$			&	the number of data points in $\mathcal{O}$	\\	\hline
$\mathcal{C}(o_i)$	&	the NN-circle of $o_i \in \mathcal{O}$ \\ \hline
$\underline{x_i}$ (resp. $\underline{y_i}$) & the left (resp. lower) side of $\mathcal{C}(o_i)$ \\ \hline
$\overline{x_i}$ (resp. $\overline{y_i}$) & the right (resp. upper) side of $\mathcal{C}(o_i)$ \\ \hline
$e_l$		&	the $l$-th event	\\ \hline
$x_l$		&	the $x$-coordinate of $e_l$	\\ \hline
$I(l)$		&	the line status between $e_{l-1}$ and $e_l$ \\ \hline
$y_t$		&	the $t$-th element in a line status	\\ \hline
$\langle y_{t-1}, y_t \rangle$	&	two consecutive elements in a line status \\ \hline
$r_l^t$		&	the rectangle $[x_{l-1},x_l,y_{t-1},y_t]$ \\ \hline
$\mathcal{R}(\cdot)$	&	the RNN set of an object \\ 
\bottomrule
\end{tabular}
\shrink
\end{table}

\subsection{Preliminaries} \label{subsec:prelimi}

We consider two types of RNN queries: the bichromatic and monochromatic RNN queries.
In the former type, the data points and their NNs belong to two different sets $\mathcal{O}$ and $\mathcal{F}$.
In the latter type, they are from the same set, i.e., $\mathcal{O} = \mathcal{F}$.
Let $d(p,q)$ be the distance between two points $p$ and $q$.
We consider three different distance metrics: $L_{\infty}$, $L_1$, and $L_2$.
We start with solving the bichromatic RNNs with $L_{\infty}$ metric because the bichromatic type is generic and $L_{\infty}$ is simpler.

\textbf{RNN Query}.
In bichromatic RNNs, we are given two sets $\mathcal{O}$ and $\mathcal{F}$.
The set $\mathcal{O}$ can be considered as (the locations of) clients while $\mathcal{F}$ as (the locations of) facilities.
The clients find their NNs from the facility set.
The RNN set of a point $f$ in $\mathcal{F}$, denoted by $\mathcal{R}(f)$, consists of the points in $\mathcal{O}$ that have $f$ as their NN, i.e.,
$$\mathcal{R}(f) = \{ o \in \mathcal{O} \setsep \forall f' \in \mathcal{F} : d(o,f) \leq d(o,f') \}.$$
For a point $q$ not in $\mathcal{F}$, we obtain its RNN set $\mathcal{R}(q)$ by adding $q$ into the facility set $\mathcal{F}$ and computing $\mathcal{R}(q)$ as above.

\textbf{Nearest Neighbor Circle (NN-circle)}.
An \emph{NN-circle} of a point $o$, denoted by $\mathcal{C}(o)$, is a circle with $o$ being the center and the distance from $o$ to its NN being the radius. 
With the $L_{\infty}$ distance metric, the distance between two points is the maximum difference between their coordinates among all dimensions, i.e., $d(p,q) = \max\{|p_x - q_x|, |p_y - p_y|\}$ in a two-dimensional space, where the subscripts denote the coordinates in the $x$ and $y$ dimensions, respectively.
Hence, NN-circles are of a square shape.
(With $L_1$ and $L_2$, the NN-circles are of diamond and circular shapes, respectively.)
\begin{figure}[!t] 
\begin{minipage}[b]{0.48\linewidth}
\centering
\begin{overpic}[scale=.6]{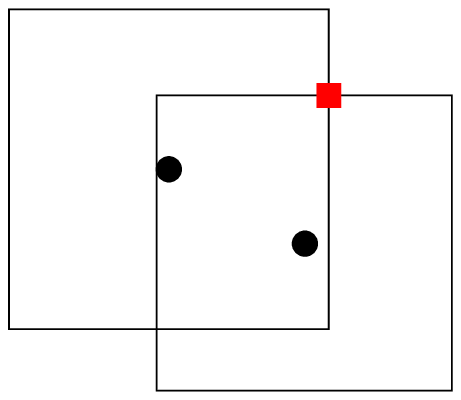}
\put(10,15){$\mathcal{C}(o_1)$}
\put(65,5){$\mathcal{C}(o_2)$}
\put(28,50){$o_1$} 
\put(68,35){$o_2$}
\put(70,70){$f_1$}
\end{overpic}
\vspace{-8pt}
\caption{NN-circles}
\label{fig:squareshape}
\end{minipage}
\begin{minipage}[b]{0.48\linewidth} 
\centering
\begin{overpic}[scale=.6]{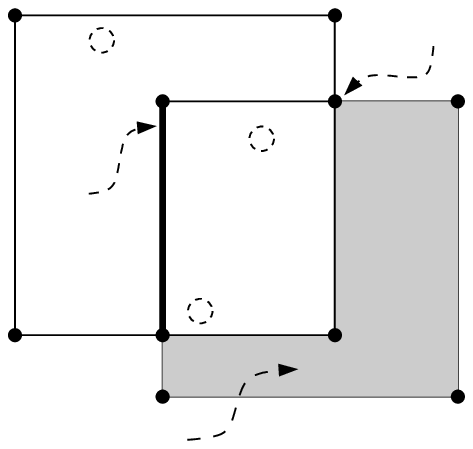}
\put(49,29){$q_1$}
\put(57,53){$q_2$}
\put(32,70){$q_3$}
\put(10,15){$\mathcal{C}(o_1)$}
\put(65,5){$\mathcal{C}(o_2)$}
\put(30,4){face}
\put(16,42){edge}
\put(71,77){vertex}
\end{overpic}
\vspace{-8pt}
\caption{Arrangement}
\label{fig:sameRNN}
\end{minipage}
\vspace{-1.5em}
\end{figure}
For example, in Fig.~\ref{fig:squareshape}, set $\mathcal{O}$ consists of two points $o_1$ and $o_2$. Set $\mathcal{F}$ consists of one point $f_1$. The NNs of $o_1$ and $o_2$ are both $f_1$.
The NN-circles $\mathcal{C}(o_1)$ and $\mathcal{C}(o_2)$ are the two squares.

\subsection{Problem Reduction} \label{subsec:equiPro}

Our goal is to draw the heat map of a given space based on the RNN sets of the points. In a continuous space, the number of points is infinite, which makes finding the RNN set of every point infeasible.
To overcome this, we reduce RNNHM to an equivalent problem Region Coloring (RC), which divides the space into regions and ``colors" (i.e., associates) every \emph{region} with a ``heat" (i.e., influence value).
We divide the space using the NN-circles as follows.
The \emph{arrangement} (i.e., layout) of the NN-circles, as illustrated in Fig.~\ref{fig:sameRNN}, forms a planar graph, which also induces a \emph{subdivision} of the space.
We use the notions in planar graphs such as \emph{vertices}, \emph{edges} and \emph{faces} (as illustrated in Fig.~\ref{fig:sameRNN}) in the arrangement directly.
In the arrangement, each face represents a unique \emph{region}, which is a maximal connected subset of the space that does not contain a vertex or an edge (e.g., the gray region in Fig.~\ref{fig:sameRNN}).
In each region, all the points have the same RNN set. 
If two points of a region have different RNN sets, there must exist at least one NN-circle that one point lies inside but the other does not; this means one side of the NN-circle must \emph{cut} the region, making it no longer a region by definition.
The RNN set of each point in the region consists of the centers of the NN-circles that \emph{enclose} the region.
For example, in Fig.~\ref{fig:sameRNN}, points $q_1$ and $q_2$ lie in the region enclosed by NN-circles $\mathcal{C}(o_1)$ and $\mathcal{C}(o_2)$, and they have the same RNN set $\set{o_1, o_2}$. For point $q_3$, its RNN set is $\set{o_1}$, which is different from that of $q_1$ or $q_2$. Therefore, $q_3$ must lie in a different region. Note that the opposite does not hold, i.e., different regions may have the same RNN set.
We formalize the above facts with the following proposition.
\begin{proposition}
The points in the same region of the subdivision formed by the arrangement of the NN-circles have the same RNN set.
\end{proposition}
For RNNHM, each region can be used to represent all the points it contains.
To associate each point with a heat, it suffices to color each region with the heat of the points it contains.
Since the influence is computed straightforwardly based on the RNN set, in the following discussion, we do not distinguish the process of outputting the RNN set of a region and the process of computing and outputting the influence value. We will simply use the term ``labeling a region'' to denote the two processes.
Assuming that the NN-circles are already precomputed (there are efficient algorithms to compute and maintain the NN-circles~\cite{Korn2000}), we define the above region coloring problem as follows.
\begin{definition}[Region Coloring]
Given a set of NN-circles, Region Coloring is to label each region in the arrangement of the NN-circles based on the RNN set of any point contained in the region.
\end{definition}

\section{A Baseline Algorithm} \label{sec:enlighten}

\begin{figure}[!t] \label{fig:enlightening}
\begin{minipage}[b]{0.49\linewidth}
\centering
\begin{overpic}[scale=.58]{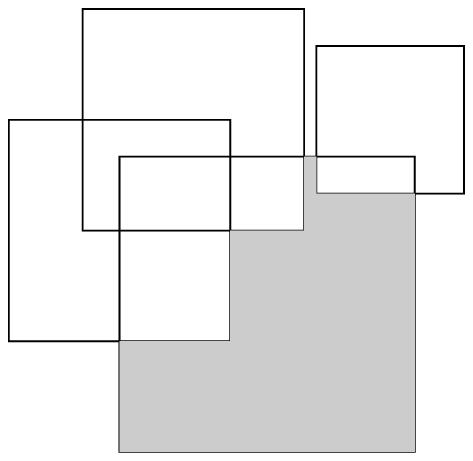}
\end{overpic}
\vspace{-8pt}
\caption{Bounding edge}
\label{fig:base_example}
\end{minipage}
\begin{minipage}[b]{0.49\linewidth} 
\centering
\begin{overpic}[scale=.58]{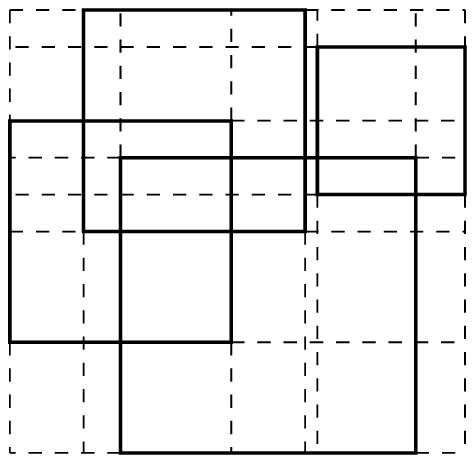}
\end{overpic}
\vspace{-8pt}
\caption{Side extension}
\label{fig:cut_base}
\end{minipage}
\vspace{-1.5em}
\end{figure}

A simple approach to the RC problem is to pick a point $p$ inside each region, use a \emph{point enclosure query} to obtain the NN-circles that enclose $p$, obtain the RNN set and label the region.
%
However, picking a point in each region is an expensive operation.
This is because it requires computing an exact representation of each region in the arrangement, which means every edge bounding a region needs to be computed (cf. Fig.~\ref{fig:base_example}) and hence has a very high complexity ($\mathcal{O}(n^2 \log n)$~\cite{de2008computational}).
To avoid such complicated computations, we extend the sides of each NN-circle to let them span across the whole arrangement, as shown in Fig.~\ref{fig:cut_base}.
By doing so we form a grid over the arrangement, where each grid cell can be easily located.
We scan the grid cells and compute the RNN set for the centroid of each cell, which solves the RC problem.
An alternative way is to use a regular grid where each cell has the same size. However, it is difficult to determine a proper cell size to guarantee that each cell falls in exactly one region unless each point is treated as a cell, which again is impossible to compute.
To efficiently compute the RNN set of a point, instead of checking each NN-circle to test whether it encloses a certain point, we build an index that supports point enclosure queries for the NN-circles. We use the \emph{S-tree}~\cite{Vaishnavi1982} for ease of analysis, although other spatial indexes such as the R-tree may be used.
\eat{Algorithm~\ref{alg:enlightening} summarizes the baseline algorithm. 
\begin{algorithm}[htb]
\DontPrintSemicolon
\small
\KwIn{An arrangement of $n$ NN-circles.}
\KwOut{A subdivision of the space with each region labeled.}
Construct an index $S$ on the NN-circles that supports point enclosure queries.\;
Sort the horizontal and vertical sides of the NN-circles.\;
Cut the regions into $m$ rectangular regions with extensions of each sides of the NN-circles.\;
\For{\emph{each of the $m$ rectangular regions}}{
	$c$ $\gets$ get the centroid of the region.\;
	$\mathcal{R}(c) \leftarrow$ get the NN-circles that enclose point $c$ with $S$.\;
	Label the rectangular region with $\mathcal{R}(c)$.\;
	}
\caption{Baseline algorithm} \label{alg:enlightening}
\end{algorithm}
}

\textbf{Algorithm Complexity}.
Let $n = |\mathcal{O}|$ denote the number of NN-circles, and $m$ denote the number of grid cells.
There are at most $2n$ extended sides vertically or horizontally, thus $m = O((2n)^2) = O(n^2)$.
To obtain the grid cells, it takes $O(2\times2n\log{2n}) = O(n\log{n})$ time to sort the sides. It then takes $O(n\log^2{n})$ time to build an S-tree index and $O(\log{n} + \alpha)$ time to process a point enclosure query~\cite{Vaishnavi1982}, where $\alpha$ is the number of NN-circles returned.
Let $\lambda$ be the maximum size of the RNN sets in the arrangement.
The time complexity of the baseline algorithm is $O(n\log^2{n} + m\log{n} + m\lambda)$. 
Since we consider a general influence measure, which can be any function with any computational cost, in the analysis we only count the number of times of influence computation, i.e., $m$ in the above complexity.
We further derive a bound for $m$ as follows. Let $r$ be the number of regions formed by $n$ NN-circles. It can be proved by \emph{Euler characteristic} that $r$ is between $\Theta(n)$ and $\Theta(n^2)$. In particular, when the $n$ NN-circles do not intersect with each other, $r = n+1 = \Theta(n)$; when the $n$ NN-circles are placed as shown in Fig.~\ref{fig:worst_case}, where they all have the same side length $n$ and the $i^{th}$ NN-circle is centered at point $(i,i)$, $r = n^2 - n + 2 = \Theta(n^2)$. Since $r \leq m$ and $m = O(n^2)$, we obtain $\Theta(r) \leq m \leq \Theta (n^2)$.

The S-tree index for point enclosure queries occupies the most space, which is
$O(n\log^2{n})$~\cite{Vaishnavi1982}.
Thus, the space complexity is $O(n\log^2{n})$.
We summarize the above analysis with the following theorem.
\begin{theorem}
The baseline algorithm for the RC problem stops in $O(n\log^2{n} + m\log{n} + m\lambda)$ time and uses $O(n\log^2{n})$ space, where $m$ is the number of grid cells and $\lambda$ is the maximum size of the RNN sets.
\end{theorem}

\textbf{Limitations of the algorithm}.
One drawback of the baseline algorithm is that it needs to process point enclosure queries, which is reflected in the $n\log^2{n}$ and $m\log{n}$ terms in the time complexity.
Another drawback is that it further divides the regions into multiple grid cells, which means a region in the original arrangement will be labeled multiple times. The number of these grid cells $m$ may increase quadratically with the increase of $n$ (closer to $\Theta(n^2)$).
A large $m$ means we need to process a large number of point enclosure queries and label a large number of grid cells, which significantly deteriorates the efficiency.
We aim to reduce $m$ (the number of times of region labeling) to the number of regions in the arrangement, which is optimal in the RC problem.
Therefore, we have two directions for improvements: (i) to avoid point enclosure queries, and (ii) to reduce the number of times region labeling.
We present our CREST algorithm which achieves these two goals in the following section.

\begin{figure}[!t]
\begin{minipage}[b]{0.49\linewidth}
\centering
\begin{overpic}[scale=.11]{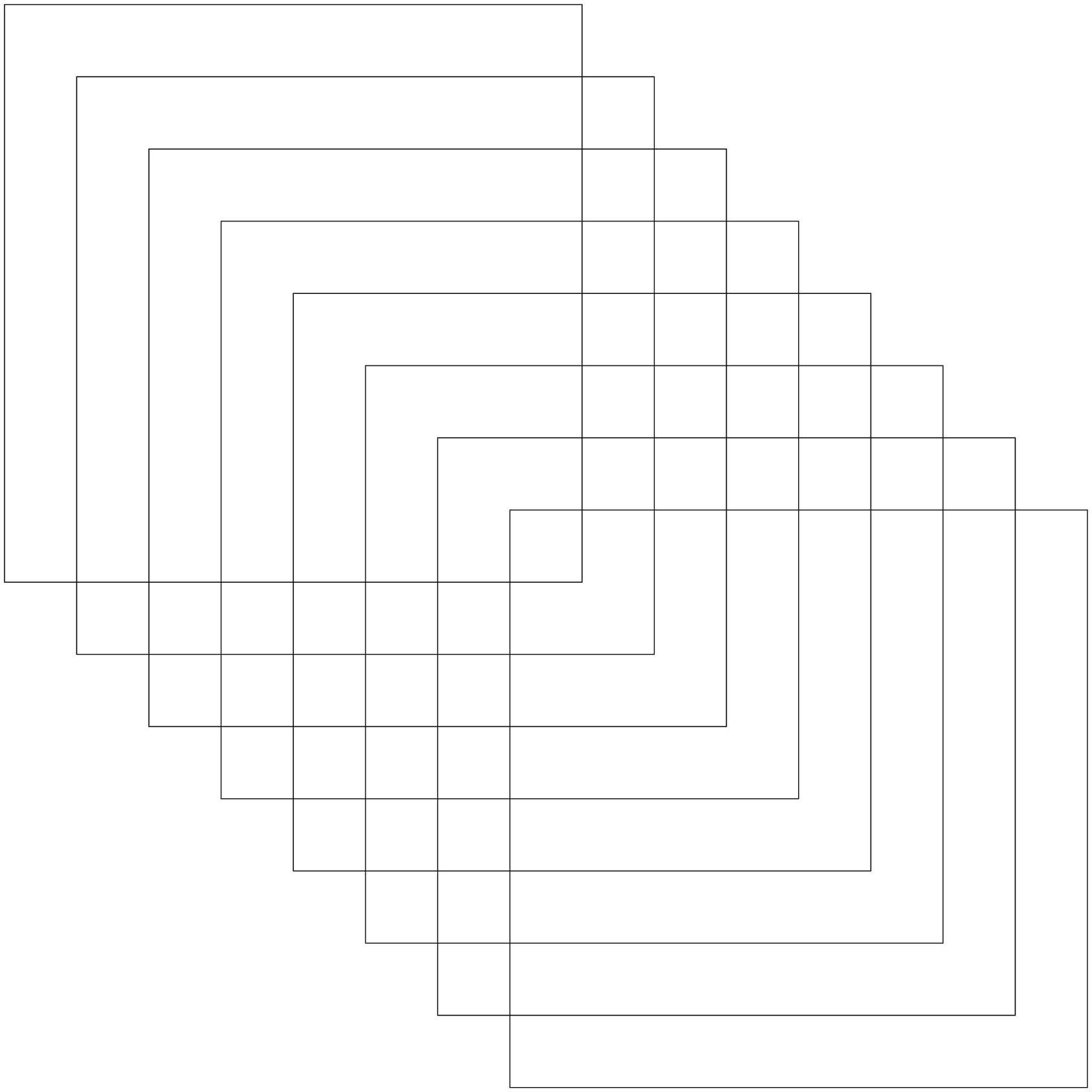}
\end{overpic}
\vspace{-8pt}
\caption{Worst case}
\label{fig:worst_case}
\end{minipage}
\begin{minipage}[b]{0.49\linewidth} 
\centering
\begin{overpic}[scale=.55]{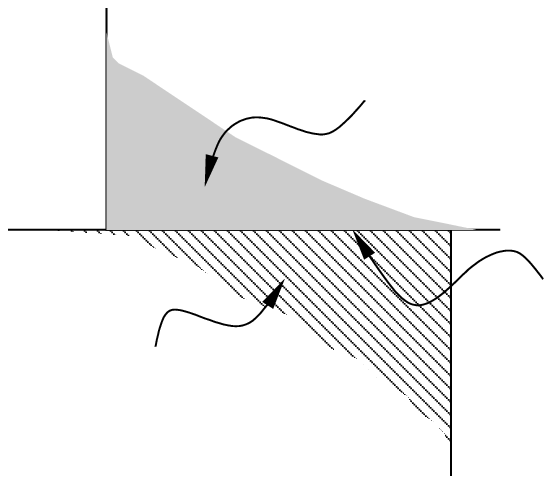}
\put(10,15){$\set{o_1,o_2}$} \put(90,30){$\ov{y_2},\un{y_3}$} \put(65,70){$\set{o_1,o_3}$}
\end{overpic}
\vspace{-8pt}
\caption{Set changing}
\label{fig:adjacent}
\end{minipage}
\vspace{-1.5em}
\end{figure}

\section{The CREST Algorithm} \label{sec:sweepalgo}

We employ the classic \emph{sweep line} strategy~\cite{Bentley1979,de2008computational} (cf. Section~\ref{sec:relatedwork}) to avoid forming a large number of cells to be labeled as done by the grid dividing strategy.
We let a line sweep from the left to the right of the space, and store information about the NN-circles that are currently \emph{cut} by the sweep line.
We call such information the \emph{line status}, and say that an \emph{event} is triggered when the line status changes.
As illustrated in Fig.~\ref{fig:sweep}, we use the \emph{distinct} vertical sides of the NN-circles as event points (i.e., $x_1,x_2,\ldots,x_9$), and the \emph{sorted} horizontal sides of the NN-circles as the line status.
Every pair of adjacent vertical sides and horizontal sides forms a \emph{subregion} to be labeled.
We notice that some of these subregions come from the same original region formed by the NN-circles, and hence do not require the RNN set and influence computations repetitively.
We use the \emph{change intervals} to avoid labeling such regions multiple times.
We avoid the RNN computation with point enclosure queries by utilizing the fact that the RNN set of a region can be obtained efficiently by modifying the RNN sets of the adjacent regions.
For example, in Fig.~\ref{fig:adjacent}, if the RNN set of the lower region is $\set{o_1,o_2}$ and the boundary between the two regions is formed by the upper and lower sides of $\mathcal{C}(o_2)$ and $\mathcal{C}(o_3)$, denoted by $\ov{y_2}$ and $\un{y_3}$, respectively, then we can immediately obtain the RNN set $\set{o_1,o_3}$ of the upper region by removing $o_2$ from $\set{o_1,o_2}$ and then adding $o_3$ to $\set{o_1}$.
We call the already-computed RNN sets of adjacent regions the \emph{base sets} and cache them for obtaining the RNN sets of newly swept regions.
We also devise techniques to constrain the number of cached base sets.
Powered by these techniques, we achieve a highly efficient algorithm to the RNNHM/RC problem, which is proved to be asymptotically optimal in many cases (cf. Section~\ref{sec:analysis}).
Next we detailed these techniques.

\begin{figure}[tb]
\centering
\begin{overpic}[scale=.9]{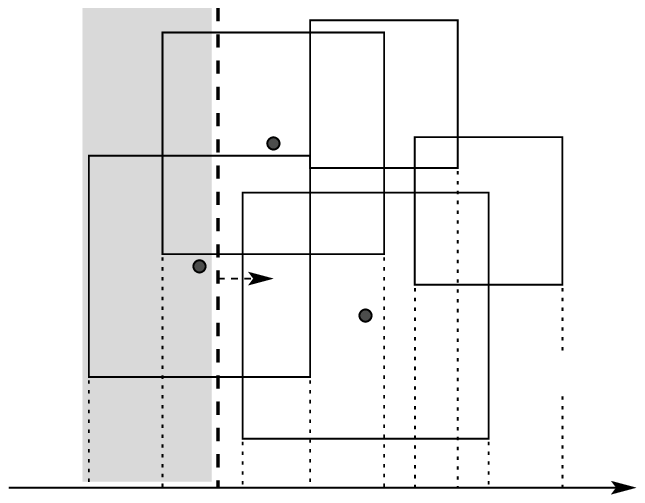}
\put(15,4){$x_1$} \put(26,4){$x_2$} \put(38,4){$x_3$} \put(47,4){$x_4$} \put(57,4){$x_5$} \put(62.5,4){$x_6$}
\put(68.5,4){$x_7$} \put(74,4){$x_8$} \put(84,4){$x_9$} \put(30,9){$s_x$} \put(30,26){$\un{y_1}$} \put(30,43.5){$\un{y_2}$} \put(30,57){$\ov{y_1}$} \put(30,75){$\ov{y_2}$}
\put(4,37){$\mathcal{C}(o_1)$} \put(15,62){$\mathcal{C}(o_2)$} \put(76,23){$\mathcal{C}(o_3)$}
\put(29,35.5){$o_1$} \put(38.5,59){$o_2$} \put(54,28){$o_3$}
\put(44.5,17){$\un{y_3}$} \put(44.5,26){$\un{y_1}$} \put(44.5,43.5){$\un{y_2}$} \put(44.5,51){$\ov{y_3}$} \put(44.5,57){$\ov{y_1}$} \put(44.5,75){$\ov{y_2}$}
\end{overpic}
\shrink
\caption{Example of events and line status}
\label{fig:sweep}
\shrink
\vspace{-5pt}
\end{figure}

\subsection{Concepts and Notation} \label{subsec:notions}


\textbf{Events}.
Let $\un{x_i}$ (resp. $\un{y_i}$) and $\ov{x_i}$ (resp. $\ov{y_i}$) be the $x$- (resp. $y$-) coordinates of the left and right (resp. lower and upper) sides of NN-circle $\mathcal{C}(o_i)$, respectively. The \emph{distinct $x$-coordinates} of the vertical sides (of all the NN-circles) are stored in ascending order in a queue, which is called the \emph{event queue} and denoted by $\mathcal{Q}_e$. The elements in $\mathcal{Q}_e$ are called the \emph{events} or \emph{event points}.
For convenience, we refer to the coordinates of the sides simply as sides when the context is clear.
We denote the $l^{th}$ event (i.e., the $l^{th}$ ejected element from $\mathcal{Q}_e$) by $e_l$ and the $x$-coordinate of $e_l$ by $x_l$.
Note the difference between an event coordinate $x_i$ and the side coordinates of $\mathcal{C}(o_i)$ (i.e., $\ov{x_i}$ or $\un{x_i}$).
They may have an equal value, but different semantics.

\textbf{Line Statuses}.
Let $s_x$ be the x-coordinate of the sweep line. We say that the line \emph{cuts} NN-circle $\mathcal{C}(o_i)$ if and only if $s_x \in (\underline{x_i}, \overline{x_i}]$ (i.e., it is in the horizontal range of $\mathcal{C}(o_i)$).
By definition, the NN-circles that are cut by the line remain the same between two consecutive events (including their positions).
Let $\mathcal{C}(o_{l_1}), \mathcal{C}(o_{l_2}), \ldots, \mathcal{C}(o_{l_{n(l)}})$ be the $n(l)$ NN-circles cut by the line when it sweeps from $e_{l-1}$ to $e_l$.
We sort the horizontal sides (not only coordinates) $\un{y_{l_1}}, \ov{y_{l_1}}$, $\un{y_{l_2}}, \ov{y_{l_2}}$, $\ldots$, $\un{y_{l_{n(l)}}}, \ov{y_{l_{n(l)}}}$ of these NN-circles in ascending order (ties are broken arbitrarily), and use the sorted list as the line status between events $e_{l-1}$ and $e_l$, which is denoted by $\mathcal{I}(l) = \|\un{y_{l_a}}, \ldots, \ov{y_{l_b}}, \ldots \un{y_{l_c}}, \ldots, \ov{y_{l_d}}\|, \; a,b,c,d \in \{1, 2, \ldots, n(l)\}$.
For example, in Fig.~\ref{fig:sweep}, the current line status is $\mathcal{I}(3) = \| \un{y_1}, \un{y_2}, \ov{y_1}, \ov{y_2} \|$.
For convenience, we denote by $y_i$ the $i^{th}$ element in the line status, and hence the line status between $e_{l-1}$ and $e_l$ is
$$\mathcal{I}(l) = \| y_1, y_2, \ldots, y_{2n(l)} \|, \;\; y_1 \leq y_2 \leq \ldots \leq y_{2n(l)}.$$

\textbf{Pair and Subregion}.
Any two consecutive elements in the line status is termed as a \emph{pair}, which is denoted by $\langle y_{t-1}, y_t \rangle$.
We denote by $\langle y_{t-1}, y_t \rangle \in \mathcal{I}(l)$ that the pair $\langle y_{t-1}, y_t \rangle$ comes from the line status $\mathcal{I}(l)$.
We denote by $[x,x',y,y']$ a rectangle whose diagonally opposite corners are $(x,y)$ and $(x',y')$ with $x < x'$ and $y \leq y'$.
When $y = y'$, $[x,x',y,y']$ is in fact a horizontal line segment. For ease of presentation, we treat it as a special rectangle.
We denote by $(x_0,y_0) \in [x,y,x',y']$ that point $(x_0,y_0)$ is in rectangle $[x,y,x',y']$.
Here the rectangle is \emph{open} (i.e., $(x_0,y_0) \in [x,y,x',y']$ iff $x_0 \in (x,x')$ and $y_0 \in (y,y')$) and no point is in the special rectangle.
When the line sweeps from $e_{l-1}$ to $e_l$, the x-coordinate $x_{l-1}$ of event $e_{l-1}$ is strictly less than that of $e_l$.
This forms a rectangle $[x_{l-1}, x_l, y_1, y_{2n(l)}]$ between the two events.
In this rectangle, each pair $\langle y_{t-1}, y_t \rangle \in \mathcal{I}(l)$ forms a \emph{small} rectangle $[x_{l-1}, x_l, y_{t-1}, y_t]$.
The small rectangle has no vertex or edge in it, which makes it a connected subset of a region. We call each small rectangle a \emph{subregion}, and denote by $r_l^t$ the one formed by pair $\langle y_{t-1}, y_t \rangle \in \mathcal{I}(l)$.
We denote by $\mathcal{R}(r_l^t)$ the RNN set of the points in subregion $r_l^t$, or simply by $\mathcal{R}(\tuple{y_{t-1}, y_t})$ when the line status is clear.

\subsection{Avoiding Point Enclosure Queries} \label{subsec:avoidpeq}

We obtain the RNN set of each subregion by finding the NN-circles enclosing it.
When the line sweeps from $e_{l-1}$ to $e_l$, the subregions between $e_{l-1}$ and $e_l$ are enclosed by the NN-circles in the $x$ dimension if and only if these NN-circles are cut by the line.
Therefore, we only need to check whether these NN-circles enclose the subregions in the $y$ dimension, which can be easily achieved by checking the line status.
We use the following lemma to show the RNN set of a pair in the line status.
Due to space limitation, we omit the proofs of the lemmas in this section.
\begin{lemma}
\label{lem:RNN}
$\forall \langle y_{t-1}, y_t \rangle \in \mathcal{I}(l)$, the RNN set $\mathcal{R}(r_l^t)$ of subregion $r_l^t = [x_{l-1}, x_l, y_{t-1}, y_t]$ is an empty set if $y_{t-1} = y_t$ or a set consists of the centers of the NN-circles that are cut by the line and enclose $r_l^t$ in the $y$ dimension, i.e., $\mathcal{R}(r_l^t)$ is
\[ \left\{
\begin{array}{l l}
\varnothing & \text{if }\; y_{t-1} = y_t,\\
\{o_i \setsep \un{x_i} < x_l \leq \ov{x_i} \text{ and } \un{y_i} \leq y_{t-1} < y_t \leq \ov{y_i} \} & \text{if }\; y_{t-1} \neq y_t.
\end{array}
\right. \]
\end{lemma}

By Lemma~\ref{lem:RNN}, we can obtain the RNN set $\mathcal{R}(\tuple{y_{t-1}, y_t})$ of a pair as follows.
When $y_{t-1} = y_t$, the RNN set is empty.
For convenience, we call such pairs \emph{invalid} pairs and the others (with $y_{t-1} < y_t$) \emph{valid} pairs.
For a valid pair, we check the elements in the line status in the range of
$(-\infty, y_{t-1}]$.
%
Since the elements are sorted in ascending order, $y_{t-1}$ (resp. $y_t$) of a valid pair must be the last (resp. first) element among elements of the same value.
Thus, we only need to check elements from the beginning of the line status to the first element (inclusive) of the pair.
Starting with an empty set, which is called the \emph{base set} and denoted by $\mathcal{R}$, if an element is a lower side, we add the center of the corresponding NN-circle to $\mathcal{R}$, otherwise we remove the center from $\mathcal{R}$.
When reaching the second element (exclusive) of the pair, we stop and $\mathcal{R}$ is the RNN set of the pair.
For example, in Fig.~\ref{fig:sweep}, the line status is $\mathcal{I}(3) = \|\un{y_1}, \un{y_2}, \ov{y_1}, \ov{y_2} \|$. For pair $\tuple{y_1,y_2}$, $\un{y_1}$ is the only element we encountered in the checking range and hence $\mathcal{R}(\tuple{y_1,y_2}) = \set{o_1}$.
We formally describe the above approach with the following corollary.
\begin{corollary}
\label{coro:RNN}
$\forall \in \langle y_{t-1}, y_t \rangle \in \mathcal{I}(l)$ with $y_{t-1} \neq y_t$, the RNN set $\mathcal{R}(r_l^t)$ of a subregion $r_l^t = [x_{l-1}, x_l, y_{t-1}, y_t]$ can be obtained by checking elements $y_i$ for $i=1$ to $t-1$ and maintaining the set $\mathcal{R}(r_l^t)$ as follows
\[ \left\{
\begin{array}{l l}
o_k \text{ is removed from } \mathcal{R}(r_l^t) & \text{if } y_i \text{ is } \ov{y_k},\\
o_k \text{ is added into } \mathcal{R}(r_l^t) & \text{if } y_i \text{ is } \un{y_k}.
\end{array}
\right. \]
\end{corollary}
It is easy to observe that if we have obtained the RNN set $\mathcal{R}(\tuple{y_{t-1}, y_t})$ of a valid pair, we can start from $y_t$ (which is the first element among elements of the same value) and use $\mathcal{R}(\tuple{y_{t-1}, y_t})$ as the base set for the valid pair $ \tuple{y_{t'-1}, y_{t'}}$ immediately next to it.
In this way, we can obtain the RNN set of \emph{every} valid pair (in one line status) with a single traversal of the line status.
Continuing with the above example in Fig.~\ref{fig:sweep}, for pair $\tuple{y_2,y_3}$, we use $\mathcal{R}(\tuple{y_1,y_2}) = \set{o_1}$ as base set, encounter $\un{y_2}$, add $o_2$, and stop with $\mathcal{R}(\tuple{y_2,y_3}) = \set{o_1,o_2}$. For pair $\tuple{y_3,y_4}$, we remove $o_1$ from $\set{o_1,o_2}$ and stop with $\mathcal{R}(\tuple{y_3,y_4}) = \set{o_2}$.

\subsection{Reducing the Number of Times of Region Labeling} \label{subsec:reducelabellingtimes}

\subsubsection{Locating the Change Interval}
With the above approach, we obtain the RNN sets and label the corresponding regions 
between two events $e_{l-1}$ and $e_l$.
We then move the sweep line forward across $e_l$ and label regions between $e_l$ and $e_{l+1}$. 
Crossing $e_l$, we obtain a new line status $\mathcal{I}(l+1)$.
We notice that some of the pairs in $\mathcal{I}(l)$ and $\mathcal{I}(l+1)$ represent the same regions (not subregions) even though they are formed by different NN-circles.
For example, in Fig.~\ref{fig:sweep}, between $e_2$ and $e_3$, $\mathcal{I}(3) = \|\un{y_1}, \un{y_2}, \ov{y_1}, \ov{y_2} \|$, while between $e_3$ and $e_4$, $\mathcal{I}(4) = \|\un{y_3}, \un{y_1}, \un{y_2}, \ov{y_3}, \ov{y_1}, \ov{y_2} \|$. The pair $\tuple{\un{y_2}, \ov{y_1}} \in \mathcal{I}(3)$ and new pair $\tuple{\ov{y_3}, \ov{y_1}} \in \mathcal{I}(4)$ represent the same region.
Besides new pairs, a pair also represents the same region if it exists in both $\mathcal{I}(l)$ and $\mathcal{I}(l+1)$ and the RNN sets of the pair in the two line statuses are the same (e.g., $\tuple{\ov{y_1}, \ov{y_2}}$ in the above example). 
The reason is that, by Lemma~\ref{lem:RNN}, the RNN set of a pair is changed if and only if the pair is entirely enclosed by an NN-circle that is inserted into (i.e., newly cut) or removed from (i.e., no longer cut by) the line. When the RNN set of a pair does not change, the two subregions formed by the pair must be connected (and hence represent the same region), since no side of NN-circles separates them.
To reduce the number of times of region labeling, we should avoid processing pairs representing the same regions, i.e., only \emph{some} of the newly formed pairs and the pairs that exist in both line status whose RNN sets are changed should be processed.
%
We use the following lemma to precisely locate the pairs that need to be processed when only one NN-circle is changed in (i.e., inserted into or removed from) the line.
\begin{lemma}
\label{lem:insertNdelete}
When a line status $\mathcal{I}(l)$ is changed into a new line status $\mathcal{I}(l')$ because an NN-circle $\mathcal{C}(o_c) = [\un{x_c}, \ov{x_c}, \un{y_c}, \ov{y_c}]$ is newly or no longer cut by the sweep line, i.e., $\un{y_c}$ and $\ov{y_c}$ are inserted into or removed from $\mathcal{I}(l)$, we only need to process the pairs in the following set
$$\{ \langle y_{t-1}, y_t \rangle \in I(l') \setsep \un{y_c} \leq y_{t-1} < y_t \leq \ov{y_c} \}.$$
\end{lemma}
By Lemma~\ref{lem:insertNdelete}, the pairs that need to be processed are located within a \emph{range}. We call such a range a \emph{changed interval} and denote it by $[y_{c_i}, y_{c_j}]$. Note that $y_{c_i}$ and $y_{c_j}$ are coordinate values, not line elements.
When the line triggers (i.e., crosses) an event, multiple NN-circles are inserted into or removed from the line, and hence several (initial) changed intervals are created.
We cannot process such changed intervals one by one, since they may intersect and affect each other. We need to merge the intersected changed intervals.
When merging two changed intervals, we need to be careful about the line elements that are of the same value so that no regions are labeled repeatedly.
Specifically, any two changed intervals $[y_{c_i},y_{c_j}]$ and $[y_{c_{i'}},y_{c_{j'}}]$ with $y_{c_i} \leq y_{c_{i'}}$ are merged into a new one $[y_{c_i},\max\{y_{c_j},y_{c_{j'}}\}]$ if $y_{c_j} \geq y_{c_{i'}}$.
After merging, we only need to handle separated changed intervals, which can be processed individually.
\begin{figure}
\centering
\begin{overpic}[scale=0.9]{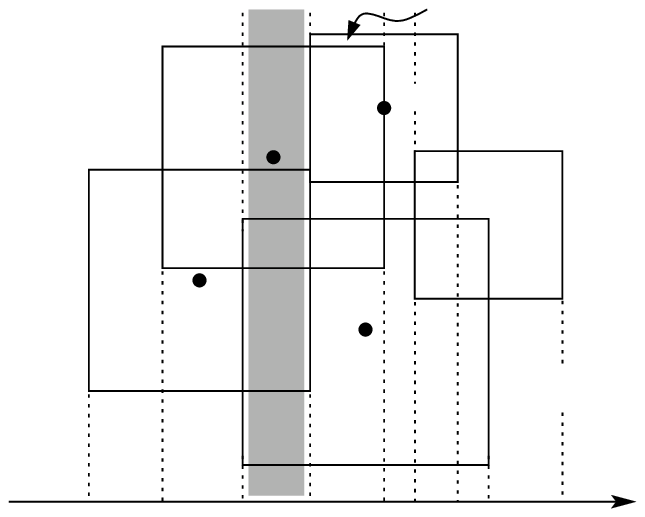}
\put(15,4){$x_1$} \put(26,4){$x_2$} \put(38,4){$x_3$} \put(47,4){$x_4$} \put(57,4){$x_5$} \put(62.5,4){$x_6$}
\put(68.5,4){$x_7$} \put(74,4){$x_8$} \put(84,4){$x_9$}
\put(43,16.5){$1$} \put(43,29.5){$2$} \put(43,42.5){$3$}
\put(60.5,49){$8$}
\put(53,43){$4$} \put(53,49){$5$} \put(53,61){$6$} \put(66,76){$7$}
\put(29,35.5){$o_1$} \put(42,60){$o_2$} \put(4.5,37){$\mathcal{C}(o_1)$} \put(15,62){$\mathcal{C}(o_2)$}
\put(62,64.5){$o_4$} \put(53,28){$o_3$} \put(71.5,64){$\mathcal{C}(o_4)$} \put(76,23){$\mathcal{C}(o_3)$}
\end{overpic}
\shrink
\caption{Example of changed intervals}
\label{fig:interval}
\shrink
\vspace{-5pt}
\end{figure}
For example, in Fig.~\ref{fig:interval}, when the line crosses $x_3$, the grey area is processed, in which $\mathcal{C}(o_3)$ is inserted into the line. The pairs (i.e., subregions), for convenience denoted by $1$, $2$, and $3$, need to be processed, which are in changed interval $[\un{y_3}, \ov{y_3}]$. When the line crosses $x_4$, $\mathcal{C}(o_1)$ is removed and $\mathcal{C}(o_4)$ is inserted. Only pairs $4$, $5$, $6$ and $7$ in $[\un{y_1}, \ov{y_4}]$ (merged from $[\un{y_1}, \ov{y_1}]$ and $[\un{y_4}, \ov{y_4}]$) are processed. When the line crosses $x_5$, $\mathcal{C}(o_2)$ is removed and only pair $8$, the only one in $[\un{y_2}, \ov{y_2}]$, is processed.
The validity of the sweep line strategy still holds with the above approach, since a simple induction will show that all pairs in all line statuses are properly processed.

\subsubsection{Caching and Retrieving Base Sets}
To obtain the RNN set within a changed interval, an efficient way is to use the RNN set of the pair that is immediately preceding of the changed interval as the base set.
Such a pair must be a valid pair, since the changed interval includes the line elements whose values are equal to the interval boundary.
Therefore, we cache (only) the RNN set of each valid pair in the line status.
We index the RNN set of a pair with its first element.
Specifically, if the pair's first element is the lower (resp. upper) side of $\mathcal{C}(o_i)$, we assign the RNN set a \emph{key} $2i-1$ (resp. $2i$).
When the RNN set of a pair is changed, the record in the index is also updated accordingly (for elements of the same value, the record is always maintained only at the last one for efficient access and space saving).
In this way, the base set for a changed interval is the record of the element that is one position ahead of the changed interval. (In case that the change interval is at the end of the line status, we also keep an empty set for the last element of a line status.)
When we process several separated changed intervals \emph{in ascending order}, it is guaranteed that such a record is always available and up-to-date.
Specifically, let $y_t$ be the element whose record we need.
If no such element $y_t$ exists, the base set is an empty set, since the changed interval must be at the beginning of the line status.
If $y_t$ is the boundary of a preceding changed interval, then such a record is already updated and ready to use, otherwise $y_t$ must exist in the last line status and the record is also available.

We use an example to illustrate the above approach.
In Fig.~\ref{fig:sweep}, $\mathcal{I}(1)$ is empty, $\mathcal{I}(2) = \| \un{y_1}, \ov{y_1} \|$, and $[\un{y_1}, \ov{y_1}]$ is the changed interval which is at the beginning of $\mathcal{I}(2)$.
We thus use an empty set as the base set, and keep the records $(2 \times 1 - 1,\set{o_1})$ for $\tuple{\un{y_1}, \ov{y_1}}$, and $(2 \times 1,\varnothing)$ for $\ov{y_1}$ (the last element), respectively.
In $\mathcal{I}(3)$, $\mathcal{C}(o_2)$ is inserted and the changed interval is $[\un{y_2}, \ov{y_2}]$. The element immediately preceding the changed interval is $\un{y_1}$.
We obtain $\set{o_1}$ as the base set with key $2 \times 1 - 1 = 1$,
and keep records $(2 \times 2 - 1 = 3, \set{o_1,o_2})$, $(2 \times 1 = 2, \set{o_2})$ and $(2 \times 2 = 4, \varnothing)$ for $\tuple{\un{y_2}, \ov{y_1}}$, $\tuple{\ov{y_1}, \ov{y_2}}$ and $\ov{y_2}$, respectively. We now have $(1,\set{o_1})$, $(2,\set{o_2})$, $(3,\set{o_1,o_2})$ and $(4,\varnothing)$ cached for future use.

\subsection{The Algorithm}

We now present the detailed steps of CREST, as summarized in Algorithm~\ref{alg:sweep}. We first obtain the event queue $\mathcal{Q}_x$ by storing the vertical sides of the NN-circles in $\mathcal{Q}_x$ in ascending order (line $4$).
The sides are stored in a way such that for each side, we can directly obtain the NN-circle to which it belongs and whether it is the left or right side.
\vspace{-5pt}
\begin{algorithm}[htb]
\DontPrintSemicolon
\small
\KwIn{An arrangement of $n$ NN-circles}
\KwOut{A subdivision with each region labeled}
$\mathcal{T} \leftarrow \varnothing$ \hfill $\diamond$ the index structure for the horizontal sides\\
$\mathcal{U} \leftarrow \varnothing$ \hfill $\diamond$ the changed NN-circles between events\\
$\mathcal{P} \leftarrow \varnothing$ \hfill $\diamond$ the cached RNN sets\\
$\mathcal{Q}_x \leftarrow$ the vertical sides of the NN-circles in ascending order\;
\For{each element $v$ in $\mathcal{Q}_x$}{
	$\mathcal{C}(o_i) \leftarrow$ the NN-circle to which $v$ belongs\;
	Add $\mathcal{C}(o_i)$ into $\mathcal{U}$\;
	\If{$v$ is a left side}{
		Insert $\un{y_i}$, $\ov{y_i}$ into structure $\mathcal{T}$\;
	}
	\Else{
		Delete $\un{y_i}$, $\ov{y_i}$ from structure $\mathcal{T}$\;
		Remove the corresponding records from $\mathcal{P}$\;
	}
	\If{the next element $v'$ of $\mathcal{Q}_x$ equals $v$}{
		Continue the outer for-loop\;
	}
	Merge the changed intervals of the NN-circles in $\mathcal{U}$\;
	Delete all elements in $\mathcal{U}$\;
	\For{\emph{each separated changed interval}}{
		 Find the starting element $st$ and ending element $ed$\; 
		 $\mathcal{R} \leftarrow$ Retrieve the base set from $\mathcal{P}$\;
		 \For{each element $y$ between $st$ and $ed$}{
		 	$\mathcal{C}(o_j) \leftarrow$ the NN-circle to which $y$ belongs\;
		 	\If{$y$ is the lower side}{
		 		$\mathcal{R} \leftarrow$ add $o_j$ into the base set\;
		 		$p \leftarrow 2j-1$
		 	}
		 	\Else{
		 		$\mathcal{R} \leftarrow$ remove $o_j$ from the base set\;
		 		$p \leftarrow 2j$
		 	}
		 	\If{$y$ is greater than the next element $y'$}{
				Label the region represented by pair $\tuple{y, y'}$ with set $\mathcal{R}$\;
				$\mathcal{P}[p] \leftarrow \mathcal{R}$\;
			}
		 }
	}
}
\caption{The CREST algorithm} \label{alg:sweep}
\end{algorithm}
\vspace{-5pt}
We then process the elements in $\mathcal{Q}_x$ one by one (line $5$). If an element is a left (resp. right) side, we insert (resp. remove) the two horizontal sizes of the NN-circle corresponding to the element into (resp. from) a balanced search tree $\mathcal{T}$ in which the data are stored in the doubly linked leaf nodes (e.g., a B$^+$-tree) (lines $6$-$14$).
When the next element in $\mathcal{Q}_x$ is greater than the current one, we process the event. The structure $\mathcal{T}$ now stores the information of the current line status.
We obtain separated changed intervals by merging the $y$-coordinates of the inserted and removed sides in this event (line $15$).
For each changed interval (line $17$), we locate the starting and ending elements in $\mathcal{T}$ (line $18$), and retrieve the base set from the RNN set records (line $19$), which are stored in a random access data structure such as an array.
We then sequentially check the elements in each changed interval (line $20$).
For each element, we either add or remove the corresponding data point (i.e., $o_i$) from the base set to obtain RNN sets of the valid pairs and label the regions (lines $21$-$29$).
To facilitate efficient insert, delete and copy operations on the base set, we keep the data points in a linked list and store pointers to the nodes in the linked list with an additional random access data structure indexed by the data points.
The RNN sets we obtained are also dynamically recorded to support the future base set retrieval (line $30$). After all changed intervals are processed, we eject the next element in $\mathcal{Q}_x$ and repeat the above steps until $\mathcal{Q}_x$ is empty. 

\section{Complexity Analysis} \label{sec:analysis}


\subsection{Complexity of CREST} \label{subsec:complexity}

We analyze the time complexity of CREST following the steps in Algorithm~\ref{alg:sweep}.
We sort the $2n$ vertical sides of the $n$ NN-circles in $O(n\log{n})$ time.
When we process the events, each horizontal side is inserted into and then deleted from the structure $\mathcal{T}$ once.
Therefore, there are at most $2n$ elements in $\mathcal{T}$, and the $2 \times 2n$ insertions and deletions can be done in $O(n\log{n})$ time.
To merge the changed intervals at an event, we can first sort them in lexicographical order and then obtain the merged result with a linear scan. This requires $O(\beta\log\beta + \beta) = O(\beta\log{\beta})$ time, where $\beta$ is the number of changed intervals at the event. Since each NN-circle can only be a changed interval twice, the total number of changed intervals in all events is $O(n)$. Thus, the overall time required for merging the changed intervals is bounded by $O(\sum\beta\log{\beta}) = O(\log{n}\sum\beta) = O(n\log{n})$.
For each merged changed interval, we obtain its starting element in $\mathcal{T}$ in $O(n\log{n} + \lambda)$ time, where $\lambda$ is the maximum size of the RNN sets in the arrangement.
This is because we first search in $\mathcal{T}$ in $O(\log{n})$ time to obtain an element $y_i$ whose value is equal to the lower endpoint of the interval. Starting from $y_i$, we obtain the starting element by checking backward (to the beginning of $\mathcal{T}$) until the elements are less than $y_i$. This procedure takes $O(2\lambda) = O(\lambda)$ time, since $\lambda$ is the maximum size of the RNN sets and there are at most $\lambda$ upper sides and $\lambda$ lower sides that are of the same $y$-coordinate.
Symmetric analysis applies to obtaining the ending element. We have only $O(n)$ changed intervals, and thus obtaining starting and ending elements can be done in $O(n\log{n} + n\lambda)$ time.
We then process the elements between them. We first retrieve a base set, which takes at most $O(\lambda)$ copying time. Thus, it takes $O(n\lambda)$ time to obtain base sets for $O(n)$ changed intervals.
For each element between the starting and ending elements, we either add into or remove from the base set its corresponding data point (i.e., $o_i$). It takes at most $O(\lambda)$ time for the adding or removing operations to obtain an RNN set for a valid pair. This is because to get an RNN set of size $\alpha_t$ by changing an RNN set of size $\alpha_s$, at most $\alpha_s$ data points are removed and $\alpha_t$ data points are added, which takes $O(\alpha_s + \alpha_t) = O(\lambda)$ time. 
We denote by $k$ the number of valid pairs, and hence the time for obtaining the RNN sets for $k$ valid pairs is bounded by $O(k\lambda)$.
For each valid pair, we record its RNN set and label its corresponding region, and this takes $O(k\lambda)$ time.

Putting all things together, we have that CREST stops in $O(n\log{n} + n\lambda + k\lambda)$ time.
Since $k$ denotes the number of times of region labeling in CREST, $k$ must be greater than or equal to the number of regions in the arrangement which is in turn greater than or equal to the number of NN-circles $n$.
Therefore, the time complexity of CREST is $O(n\log{n} + k\lambda)$. Shortly (in Section~\ref{subsec:bound_k}) we will prove that $k = \Theta(r)$, where $r$ is the number of regions in the arrangement.

The space required by the queue $\mathcal{Q}_x$ and structure $\mathcal{T}$ is $O(n)$. The space required by caching the RNN sets is $O(n\lambda)$, and the storage of the base set requires $O(n + \lambda)$ space. Overall, the space complexity of CREST is $O(n\lambda)$.

\subsection{Bounding the Times of Region Labeling in CREST}
\label{subsec:bound_k}

From the above analysis, we can see that the number of times of region labeling $k$ largely decides the performance of CREST.
In CREST, we successfully avoid labeling the same region multiple times in different line statuses. Although rarely happens, multi-labeling still exists within the same line status. For example, in Fig.~\ref{fig:repeatlabelling}, at the event of the left side of $\mathcal{C}{(o_1)}$, the gray region is labeled six times.
\begin{figure}[!t]
\begin{minipage}[b]{0.49\linewidth}
\centering
\begin{overpic}[scale=.5]{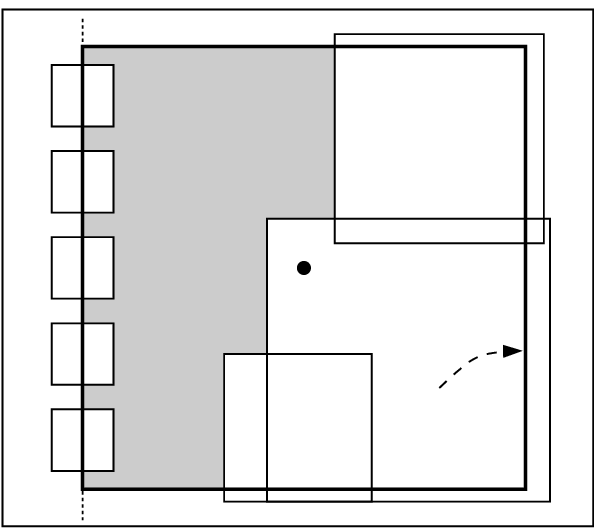}
{\small
\put(50,38){$o_1$} \put(65,15){$\mathcal{C}{(o_1)}$}
}
\end{overpic}
\vspace{-8pt}
\caption{Multilabelling}
\label{fig:repeatlabelling}
\end{minipage}
\begin{minipage}[b]{0.49\linewidth} 
\centering
\begin{overpic}[scale=0.5]{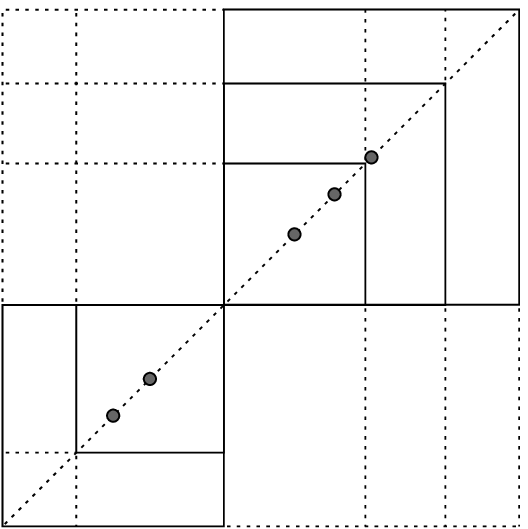}
{\small
\put(-3,105){$a_5$} \put(13,105){$a_3$} \put(40,105){$a_1$} \put(66,105){$a_4$} \put(81,105){$a_2$} \put(95,105){$a_6$}
\put(-15,-1){$a_5$} \put(-15,13){$a_3$} \put(-15,40){$a_1$} \put(-15,68){$a_4$} \put(-15,83){$a_2$} \put(-15,97){$a_6$}
\put(22,17){$o_5$} \put(31,26){$o_3$} \put(45,55){$o_4$} \put(52,63){$o_2$} \put(59,73){$o_6$}
}
\end{overpic}
\vspace{-8pt}
\caption{Reduction}
\label{fig:reduction}
\end{minipage}
\vspace{-1.5em}
\end{figure}
Despite the multi-labeling, we show with the following lemma that the number of times of region labeling in CREST and the number of regions in the arrangement, up to a constant factor, are asymptotically the same (as a function of $n$).
\begin{lemma}
$k = \Theta(r)$, where $r$ is the number of regions in the arrangement.
\end{lemma}
\begin{proof}
Let $v$, $e$, and $c$ be the number of vertices, edges and connected components in the arrangement, respectively.
We call the number of edges bounding a region the degree of the region, and the number of edges incident to a vertex the degree of the vertex.
In an arrangement of squares, there are only $2$-, $3$-, and $4$-degree vertices, which are denoted by $v_2$, $v_3$, and $v_4$, respectively.
In CREST, the number of times a region is labeled cannot be greater than the degree of the region, since each time the region is labeled we need a distinct valid pair which requires at least one of the edges bounding the region.
Therefore, $k$ is less than or equal to the sum of degrees of all regions which equals $2e$, i.e., $k \leq 2e$.
%
In the arrangement, we also have $v=v_2+v_3+v_4$, $2e = 4v_4+3v_3+2v_2$ and $v - e + r -c = 1$ (Euler characteristic). Combining these three equations, we obtain
$r = v_4 + v_3/2 + c + 1.$
We then have that
$$k \leq 2e = 4v_4+3v_3+2v_2 \leq 6(v_4 + v_3/2 + c + 1) + 2v_2 = 6r + 2v_2.$$
The number of 2-degree vertices is less than or equal to $4n$ and hence less than or equal to $4r$, since each square makes at most four 2-degree vertices and $n \leq r$. Therefore, it follows that
$$k \leq 6r + 2v_2 \leq 6r + 8n \leq 14r.$$
Obviously, $r \leq k$, and hence $r \leq k \leq 14r$, which completes the proof.
\end{proof}
We conclude the above analysis with the following theorem.
\begin{theorem} \label{theo:complexity}
The CREST algorithm solves the (bichromatic) RC problem in $O(n\log{n} + r\lambda)$ time with $O(n\,\lambda)$ space, where $r$ and $\lambda$ are the number of regions and the maximum size of the RNN sets in the arrangement, respectively.
\end{theorem}

\subsection{A Lower Bound of the RC Problem} \label{subsec:lower_bound}

We show that $\Omega(n\log{n} + r\lambda^*)$ is a lower bound of the RC problem (in the algebraic computation tree model)~\cite{BenOr1983}, where $\lambda^*$ is the \emph{average} size of RNN sets in the arrangement.
When $r\lambda^*$ is the dominating term, at least the RNN sets of all regions need to be output, the above bound is a trivial lower bound. Thus, we only need to show that it requires $\Omega(n\log{n})$ operations even without considering the output cost.
This bound is proved by the reduction from the \emph{point distinctness problem} to a special case of the RC problem.
\begin{definition}[Element Distinctness]
Given real numbers $a_1, \ldots, a_n \in \mathbb{R}$, determine whether or not there is a pair $i$, $j$ with $i \neq j$ and $a_i = a_j$.
\end{definition}
We show that the element distinctness problem can be reduced to the RC problem in linear time. For each real number $a_i$, we create a point $(a_i, a_i)$, $i=1,2,\ldots,n$ in the plane. We then build a square $\mathcal{C}(o_i)$ with point $(a_i, a_i)$, $i=2,\ldots,n$ and point $(a_1,a_1)$ being the diagonally opposite corners and $o_i$ being the center. An example of such reduction is shown in Fig.~\ref{fig:reduction}.
These squares form an arrangement of NN-circles in a two-dimensional space. We use this arrangement as input to any algorithm that solves the RC problem.
A correct algorithm outputs exactly $n$ RNN sets (including the empty set) if and only if the elements are distinct. The reason is that each RNN set corresponds to only one region, and there are $n$ regions (including the exterior face) in the arrangement if and only if the elements are distinct. It has been proved that the element distinctness problem has a lower bound $\Omega(n\log{n})$~\cite{BenOr1983} (in the algebraic computation tree model), which implies that RC has a lower bound $\Omega(n\log{n})$ without the output cost.
Therefore, $\Omega(n\log{n} + r\lambda^*)$ is a lower bound of the RC problem.

\subsection{Optimality of CREST} \label{subsec:optimality}

From the time complexity of CREST and lower bound of the RC problem, CREST is asymptotically optimal in terms of the number of times of region labeling (i.e., influence computation)
in all cases, since $k = \Theta(r)$.

In the following cases, we show that the upper bound $O(n\log{n} + r\lambda)$ of CREST is also tight, which indicates that CREST is overall asymptotically optimal.
For the bound to be tight, it is sufficient to show that $\lambda = \Theta(\lambda^*)$.

\textbf{Case (i)}. When the clients and facilities are relatively uniformly distributed such that $\lambda$ is bounded by a sufficiently large constant $C$ (which depends on $\frac{|\mathcal{O}|}{|\mathcal{F}|}$), since $\lambda^* \leq \lambda$, $\lambda^*$ is also bounded by $C$. Thus, $\lambda = \Theta(\lambda^*) = O(1)$. An example is that none of the $n$ squares intersects any other ones (or only a few of them overlap).

\textbf{Case (ii)}. When $\lambda$ is unbounded, we show with the worst case illustrated in Fig.~\ref{fig:worst_case} that $\lambda = \Theta(\lambda^*)$ also holds when every square intersects all the other ones.
In this arrangement, $\lambda = n$, and we have that $r = n^2 - n + 2$ and $r \cdot \lambda^* = \frac{n^3 + 2n}{3}$. Therefore, it follows that
$$\lambda^* = \frac{n^3 + 2n}{3(n^2 - n + 2)} = \frac{n}{3} \cdot \frac{n^3 + 2n}{n^3 - n^2 + 2n} \geq \frac{n}{3} = \frac{\lambda}{3}.$$
Since $\lambda^* \leq \lambda$, it follows that $\lambda = \Theta(\lambda^*)$, which indicates CREST is overall asymptotically optimal.

\section{RNNHM in Other Settings} \label{sec:extension}

We show how CREST solves the RNNHM problem with the monochromatic RNNs, $L_1$, and $L_2$ distance metrics.

\subsection{Monochromatic RNNs} \label{subsec:mono}

CREST directly applies to the monochromatic RNNs, since $\mathcal{O}$ and $\mathcal{F}$ being the same set does not affect the computation of the NN-circle and these NN-circles still form a planar subdivision with axis-aligned edges as in the bichromatic RNNs.
By Korn et al.~\cite{Korn2000}, an RNN set contains at most six points for monochromatic RNN queries, which means $\lambda = O(1)$.
Therefore, by Theorem~\ref{theo:complexity},
the time complexity of CREST for the monochromatic RNNs is $O(n\log{n} + r)$ and the space complexity is $O(n)$.

\subsection{RNNHM with L{\small 1} Distance} \label{subsec:L1}
In two-dimensional spaces, the $L_1$ distance can be viewed as equivalent to the $L_\infty$ distance by rotation and scaling. Specifically, with the $L_1$ distance, NN-circles are of diamond shape. If we rotate (around the origin) the coordinate system counter-clockwise by $\pi/4$, diamonds become squares. Each point $(x,y)$ in the original system has a corresponding point $(x',y')$ in the rotated system with $x' = x\cos\theta - y\sin\theta$, $y' = x\sin\theta + y\cos\theta$ and $\theta = \pi/4$.
In the rotated system, CREST directly applies.
The transformation takes $O(n)$ time and the overall time and space complexities stay unchanged.

\subsection{RNNHM with L{\small 2} Distance} \label{subsec:L2}


\begin{figure}
\centering
\begin{overpic}[scale=0.12]{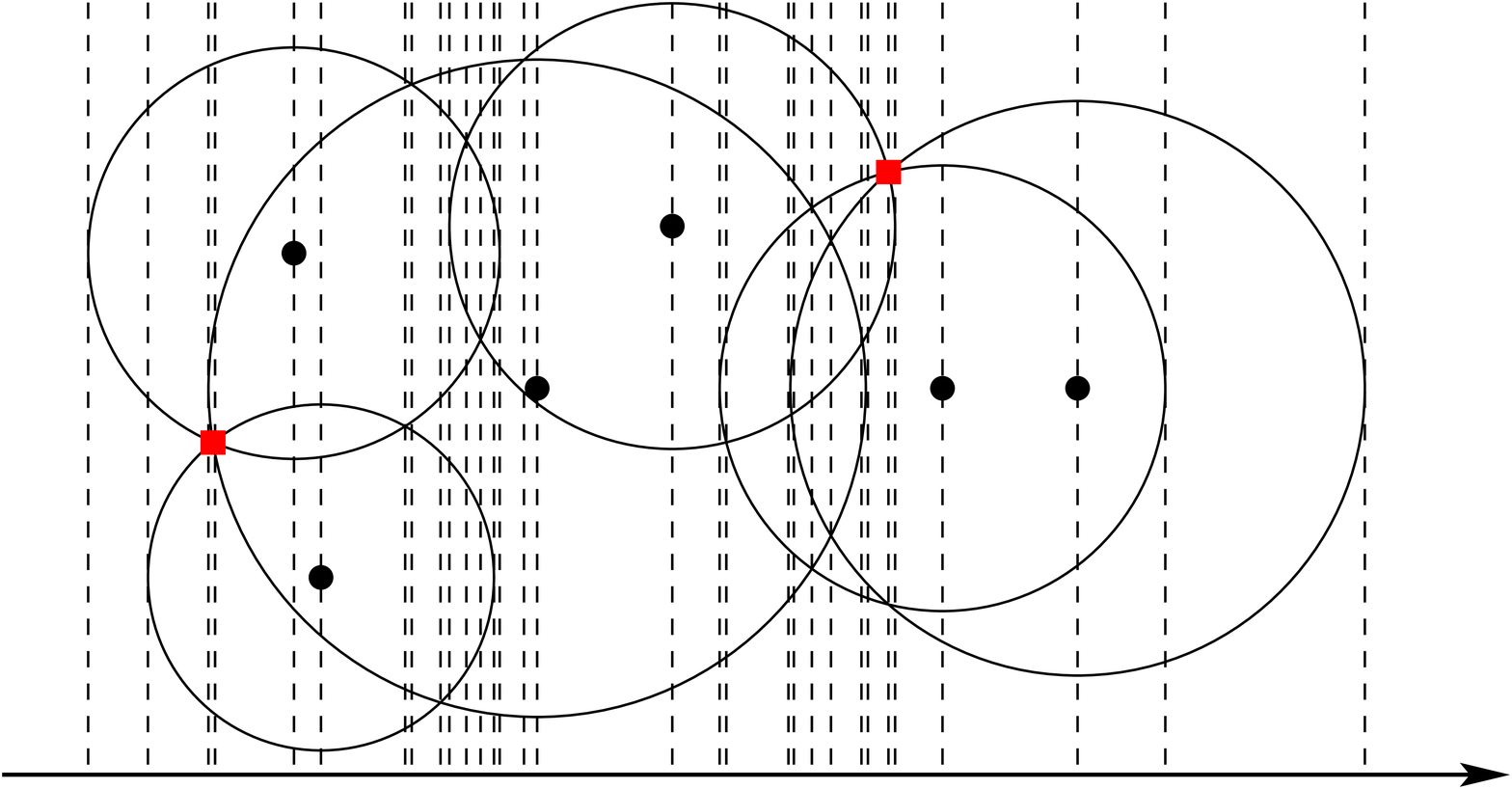}
\put(4,1){$x_1$} \put(8.5,1){$x_2$} \put(14,1){$x_4$} \put(18.5,1){$x_5$} \put(24,1){$x_7$} \put(34,1){$x_{16}$} \put(41,1){$x_{17}$} \put(49,1){$x_{20}$} \put(61,1){$x_{28}$} \put(70,1){$x_{29}$} \put(88,1){$x_{31}$}
\put(38,11){$\un{y_1}$} \put(38,23){$\un{y_2}$} \put(38,46){$\ov{y_1}$} \put(38,56){$\ov{y_2}$}
\put(37,30){$o_1$} \put(40,37){$o_2$}
\end{overpic}
\shrink
\caption{CREST with $L_2$ distance metric}
\label{fig:L2}
\vspace{-1.8em}
\end{figure}

With the $L_2$ distance, the NN-circles are of circular shape. They form a planar subdivision with \emph{curved} edges, as shown in Fig.~\ref{fig:L2}.
CREST still applies in such a subdivision but requires modifications as follows.
We use the $x$-extreme points of circles (instead of vertical sides of squares) as event points. In the line status, we use the arc segments of circles between two consecutive events as line elements (instead of horizontal sides of squares).
For each line element $y_i$ (i.e., an arc segment), we assign two values $y_i^s$ and $y_i^l$, which are the smallest and largest $y$-coordinates of $y_i$ between two consecutive events $e_{l-1}$ and $e_l$, respectively.
Line element $y_i$ is less than $y_j$ iff (i) $y_i^s < y_j^s$ or (ii) $y_i^s = y_j^s$ and $y_i^l < y_j^l$ or (iii) $y_i^s = y_j^s, y_i^l = y_j^l$ and $y_i^m < y_j^m$, where $y_i^m$ and $y_j^m$ are the y-coordinates of $y_i$ and $y_j$ at $\frac{x_{l-1} + x_l}{2}$, respectively.
We include intersection points as event points (e.g, $x_4$ and $x_7$ in Fig.~\ref{fig:L2}). This is because the arc segments of NN-circles switch positions at intersection points.
We also use the center of each NN-circle as event points (e.g., $x_5$ and $x_{29}$ in Fig.~\ref{fig:L2}) to guarantee that each line element is $y$-monotone (i.e., strictly increasing or decreasing in the $y$-dimension).
Before processing an event, we update values $y_i^s$ and $y_i^l$ for each line element $y_i$ regardless of whether it is related to the event.
This update is required in order to maintain a proper order in the line status because the arcs go up or down between events and their upper and lower values change even if they are irrelevant to the events.
Note that such an update does not change the relative order of the line elements irrelevant to the events, and thus can be completed in linear time.

Apart from the above modifications, CREST remains the same.
Specifically, if an event point at $e_{l-1}$ is the left boundary point of $\mathcal{C}(o_c)$, we insert $\un{y_c}$ and $\ov{y_c}$ into the line status with $\un{y_c}^l = \ov{y_c}^s = y_{o_c}$ and $\un{y_c}^s$ and $\ov{y_c}^l$ being the lower and upper $y$-coordinates where $e_l$ intersects $\mathcal{C}(o_c)$, respectively. We also create a change interval with $\un{y_c}$ and $\ov{y_c}$.
If an event point is an intersection point, we obtain the relevant line elements (arcs incident to the intersection point) in the line status, switch their positions and create a change interval with these elements.
If an event point is the right boundary point or center of an NN-circle, we remove or update the two line elements corresponding to the NN-circle. We do not create change intervals for either of these two types of event points, since no pair is between the removed elements and updating the line elements by the centers is only to keep them $y$-monotone.
We then merge and handle change intervals as before.

\textbf{Complexities}.
In the worst case (as shown in Fig.~\ref{fig:worst_case}), CREST runs in $O(n^3)$ time with the $L_2$ distance, since there can be as many as $O(n^2)$ events and for each event we need to update $O(n)$ line elements.
However, the worst case complexity is much lower than that of an existing algorithm~\cite{Sun2012cikm}, which suffers from an exponential running time in the worst case.
The algorithm was proposed to obtain regions with the maximum influence value, but it could be adapted to solve the RC problem if we remove its pruning techniques.
The algorithm~\cite{Sun2012cikm} follows the filter and refine paradigm by enumerating all possible regions and then checking their existence. For example, when $\mathcal{C}(o_1)$ intersects $\mathcal{C}(o_2)$ and $\mathcal{C}(o_3)$, it enumerates the regions $\hat{o_1}\hat{o_2}\hat{o_3}$, $\hat{o_1}\hat{o_2}\bar{o_3}$, $\hat{o_1}\bar{o_2}\hat{o_3}$, $\hat{o_1}\bar{o_2}\bar{o_3}$, where $\hat{o_i}$ means inside $\mathcal{C}(o_i)$ and $\bar{o_i}$ means outside $\mathcal{C}(o_i)$, and then checks whether such regions really exist.
In our experiments (in Section~\ref{sec:exp}), CREST constantly outruns the algorithm on data sets of various settings.

\section{Experiments} \label{sec:exp}

In this section, we experimentally evaluate the performance of CREST.
We use both real and synthetic data sets. Two real data sets, NYC and LA, contain points-of-interest in New York City and Los Angeles, respectively (we obtain the data sets from the authors of~\cite{Bao:2012:LPR}).
Table~\ref{tab:realdata} lists the details of the real data sets. We also generate two synthetic data sets, Uniform and Zipfian, which contain points of uniform and Zipfian distributions, respectively. The skew coefficient in Zipfian distribution is set to $0.2$.
\begin{table}[tb]
\shrink
\caption{Real Data Sets}\label{tab:realdata}
\shrink
\small
\centering
\begin{tabular}{c|c|c}
\toprule
\textbf{Name}	&	\textbf{Size}	&	\textbf{Description} \\
\midrule \midrule
NYC	&	128,547		&	points-of-interest in New York City		\\	\hline
LA	&	116,596		&	points-of-interest in Los Angeles		\\ 
\bottomrule
\end{tabular}
\shrink
\vspace{-5pt}
\end{table}
In the experiments, we use the bichromatic RNNs since the monochromatic type is just a special case of the bichromatic RNNs.
We use $L_1$ and $L_2$ distance metrics since they are used more often than $L_\infty$ in real-world scenarios, and $L_1$ and $L_\infty$ are equivalent in two-dimensional spaces.
We uniformly sample from the data sets to obtain the client set $\mathcal{O}$ and the facility set $\mathcal{F}$.
All algorithms are implemented using C++ and the experiments are conducted on a desktop computer with a $3.4$GHz Intel i$7$-$2600$ CPU and $8$GB main memory.

\subsection{Showcasing Real-World Heat Maps}

In the first set of experiments, we show the RNN heat maps for two cities: New York City and Los Angeles. For each data set NYC and LA, we uniformly sample $20,000$ points as the clients and $6,000$ points as the facilities, since in real world scenarios the number of clients is usually larger than the number of facilities. For simplicity, we measure the influence by the size of RNN sets, although any other function on the RNN sets may be used.
Fig.~\ref{fig:ny_hm} and Fig.~\ref{fig:ny_gm} (in Section~\ref{sec:intro}) show the RNN heat map and the satellite map of New York City (within latitude and longitude ranges $[40.50$, $40.95] \times [-74.15$, $-73.70]$), while Fig.~\ref{fig:LAhm} and Fig.~\ref{fig:LAmap} show the heat and satellite maps of Los Angeles (within $[33.82$, $34.17] \times [-118.47$, $-118.12]$), respectively.
Comparing the heat and satellite maps, we can see that they are closely geographically correlated as expected.
For instance, the mountain and sea area have few clients or facilities, and hence have very low heat.
We can easily explore regions of various influences to help various decision making applications such as those described in the motivating examples. If the decision maker is interested in any specific area, she can zoom in to see more details.
\begin{figure}[htb]
\centering
\subfigure[Heat map for LA]{
\epsfig{file=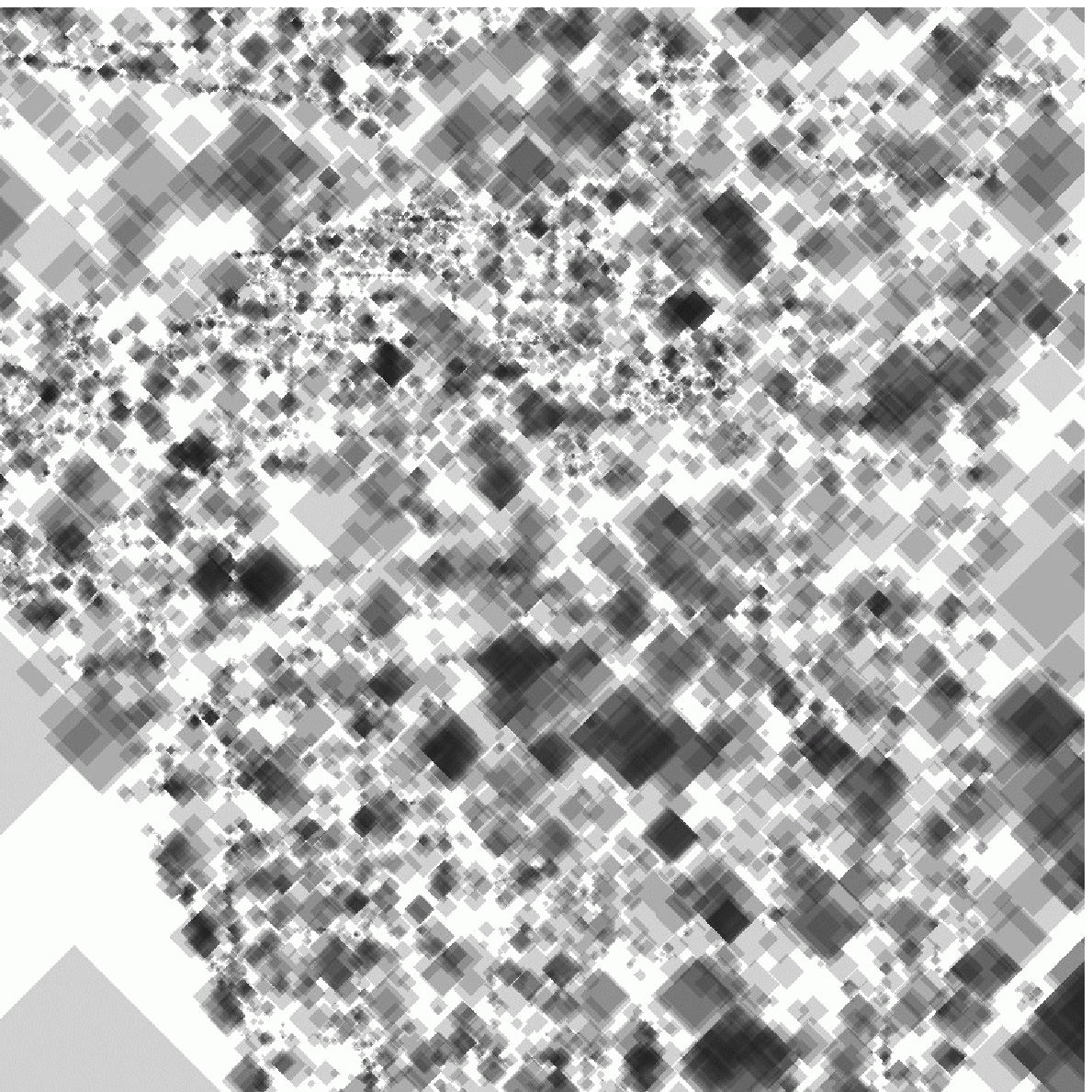, width=1.35in}
\label{fig:LAhm}
}
\subfigure[Satellite map for LA]{
\epsfig{file=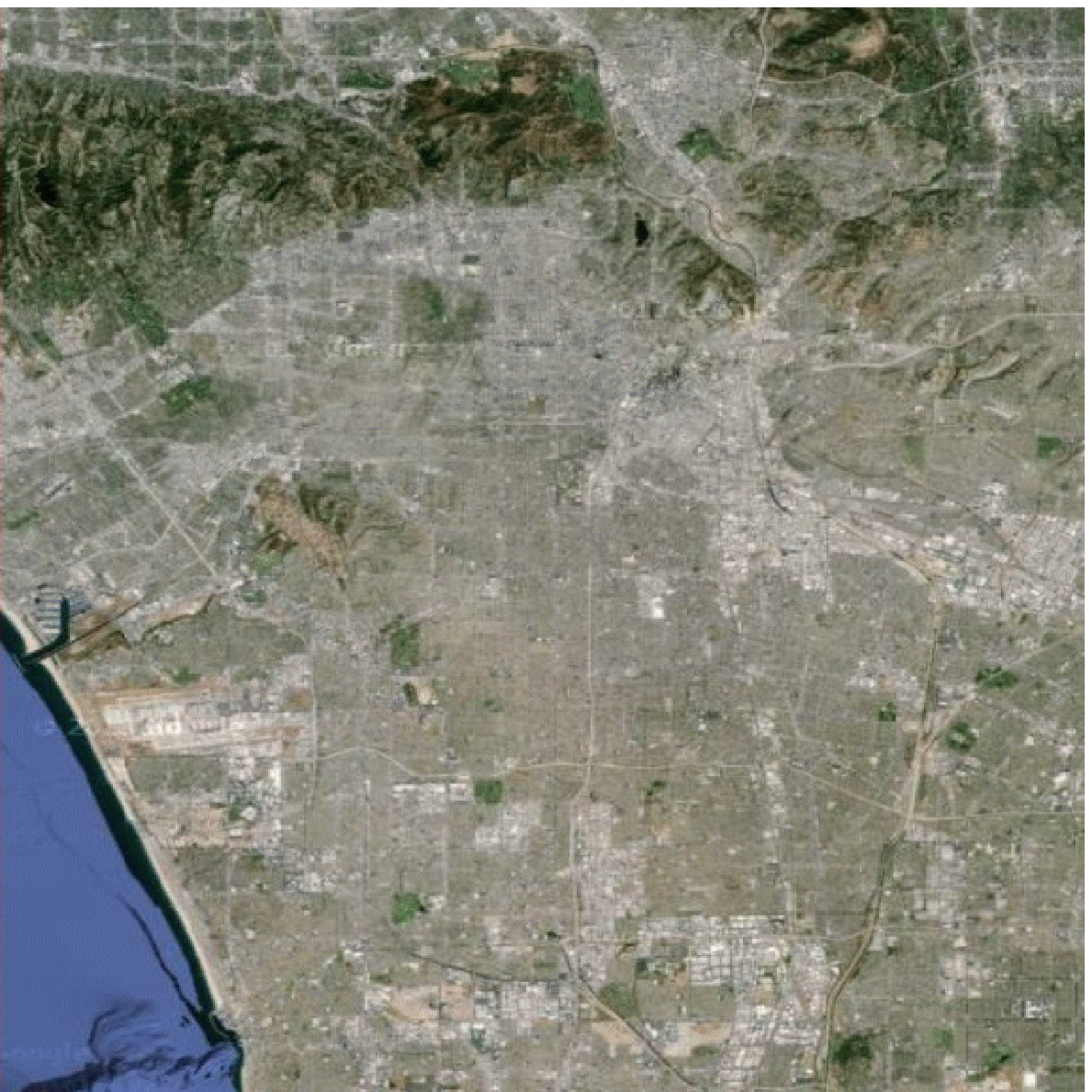, width=1.35in}
\label{fig:LAmap}
}
\vspace{-5pt}
\caption{Real-world heat map}\label{fig:hm}
\vspace{-1.5em}
\end{figure}

\subsection{Performance of CREST with L{\small 1} Distance}
In this set of experiments, we compare the running time of three algorithms: the baseline algorithm (\textbf{BA}), the CREST algorithm with only the RNN computation optimization, denoted by \textbf{CREST-A}, and the \textbf{CREST} algorithm with both RNN computation optimization and repetitive region labeling optimization. We cannot evaluate the effect of the latter optimization alone, since it is built upon the former optimization.
We compute the influence by (i) the size of RNN sets and (ii) the function considering the capacity constraints of facilities~\cite{Sun2012cikm} (described in the Introduction), respectively.
The results of the latter function are consistent with those using the size and hence are omitted due to space limitation.

\textbf{Effect of $\frac{|\mathcal{O}|}{|\mathcal{F}|}$}.
We first vary the ratio of $\frac{|\mathcal{O}|}{|\mathcal{F}|}$ from $2^1$ to $2^{10}$. 
Since the baseline algorithm does not terminate within $24$ hours on large data sets, we only show results on relatively small data sets
and fix $n = |\mathcal{O}|$ at $2^{10}$.
We plot the results in Fig.~\ref{fig:cf} (note the log scale in the axes).
We can see that in all data sets, CREST outperforms the baseline by at least three orders of magnitude, and outperforms CREST-A by several times.
With the increase of $\frac{|\mathcal{O}|}{|\mathcal{F}|}$, the running time of CREST also moderately increases. This is because both the number of regions and the maximum size of the RNN sets increase with $\frac{|\mathcal{O}|}{|\mathcal{F}|}$. The growth rates of CREST-A and CREST are similar, which indicates that the ratio of $\frac{|\mathcal{O}|}{|\mathcal{F}|}$ mainly affects the number of regions in the arrangement, but not the number of times a same region is repeatedly labeled.
We can also see from the slope of the lines that with the increase of $\frac{|\mathcal{O}|}{|\mathcal{F}|}$, the number of regions increases only polynomially (rather than exponentially). This indicates that the performance of CREST will stay stable even if $\frac{|\mathcal{O}|}{|\mathcal{F}|}$ becomes very large.

\begin{figure}[t]
\centering
\subfigure[LA]{
\epsfig{file=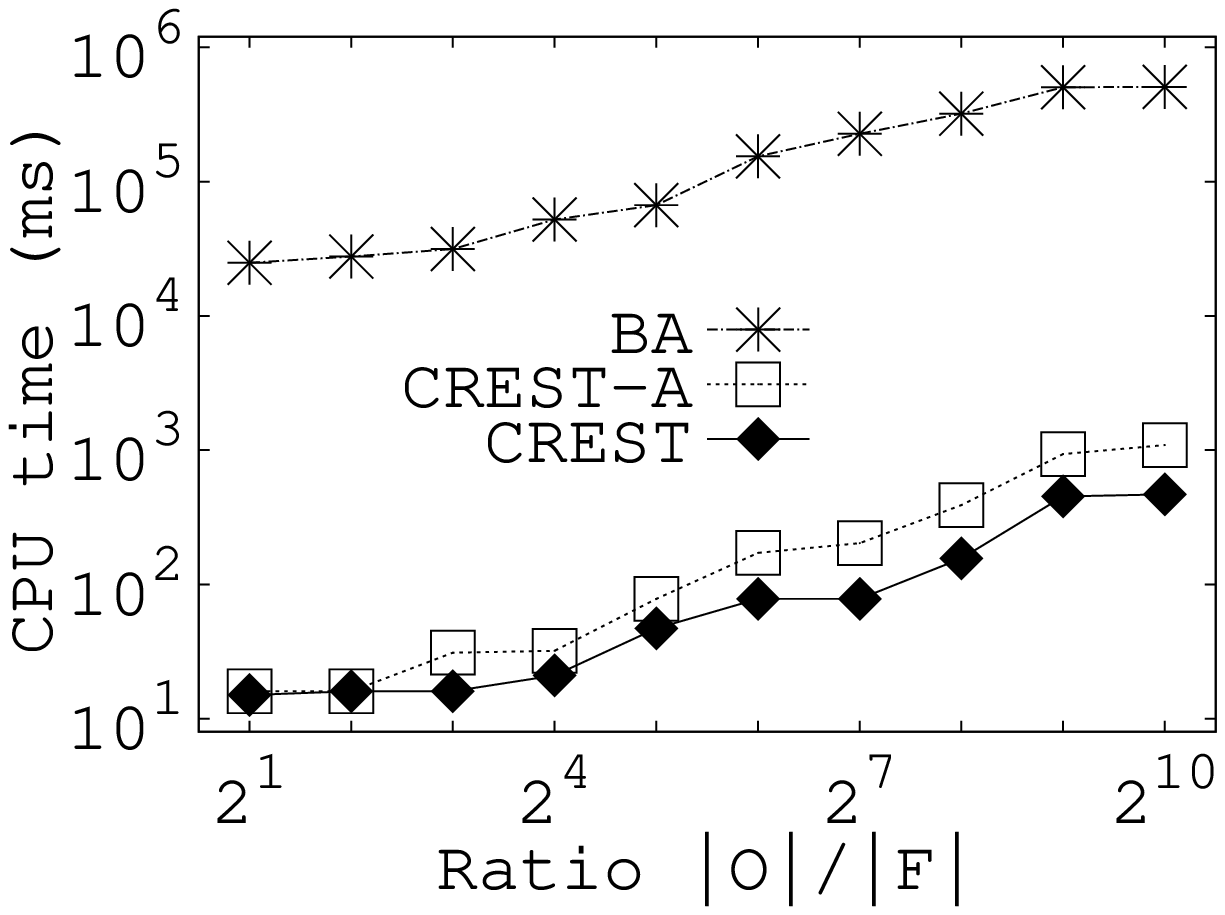, width=1.55in}
\label{fig:cfLA}
}
\subfigure[NYC]{
\epsfig{file=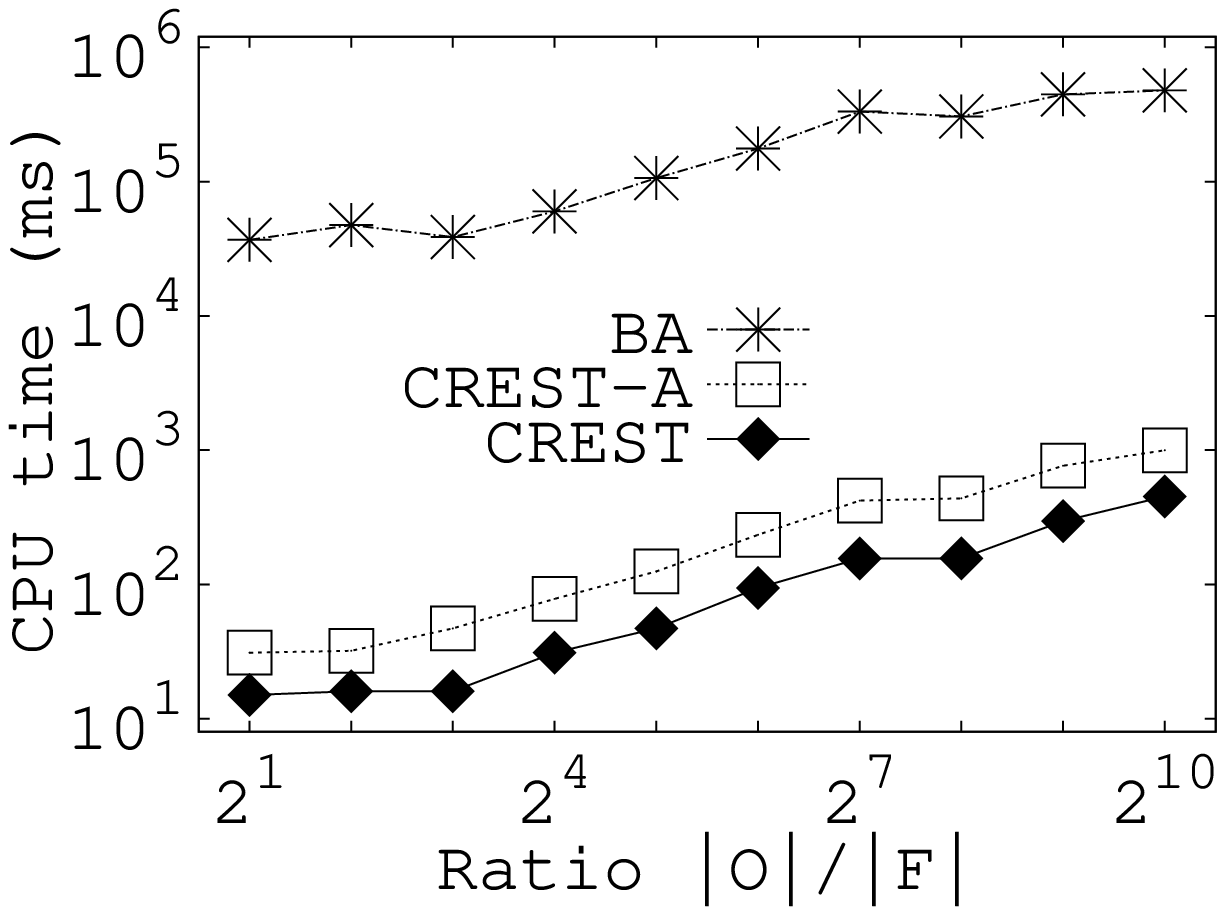, width=1.55in}
\label{fig:cfNY}
}
\subfigure[Uniform]{
\epsfig{file=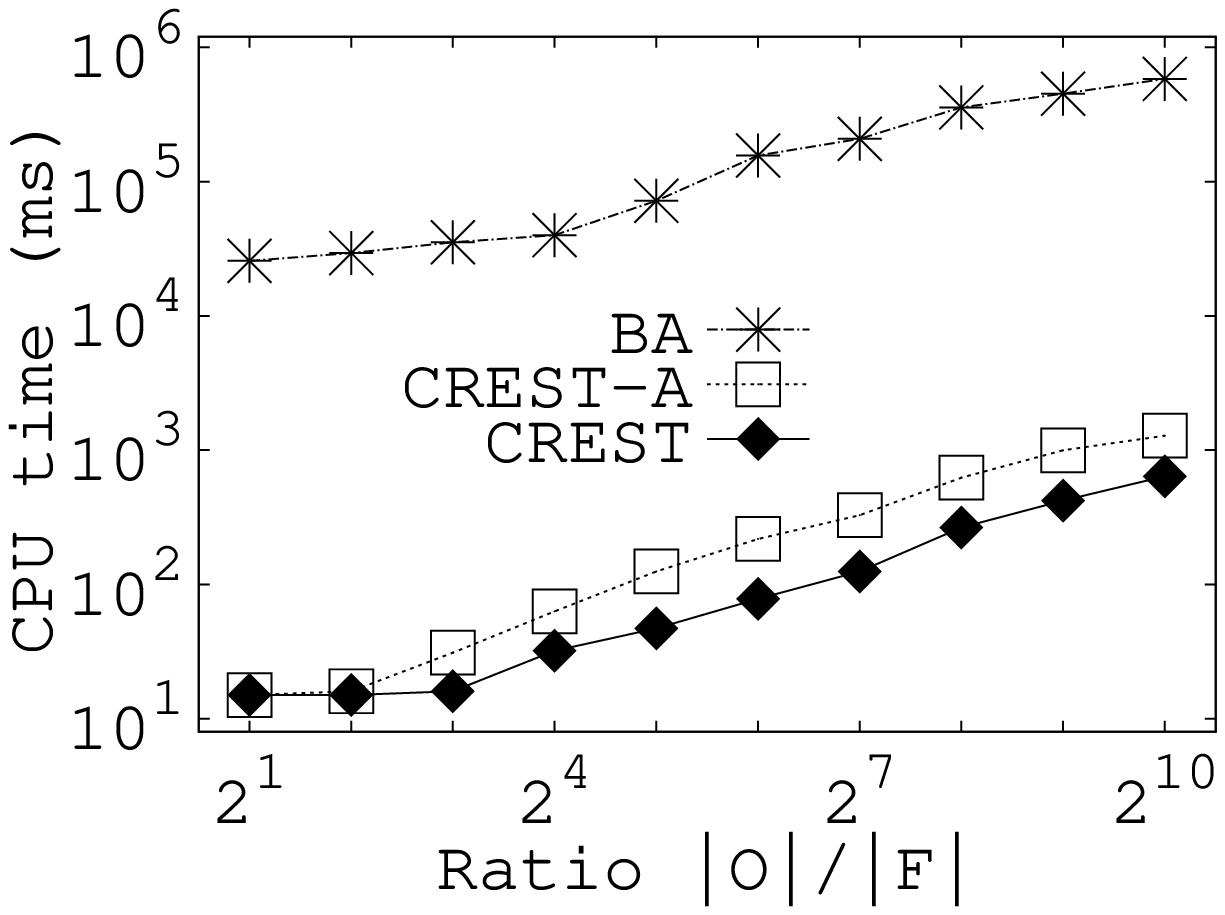, width=1.55in}
\label{fig:cfUF}
}
\subfigure[Zipfian]{
\epsfig{file=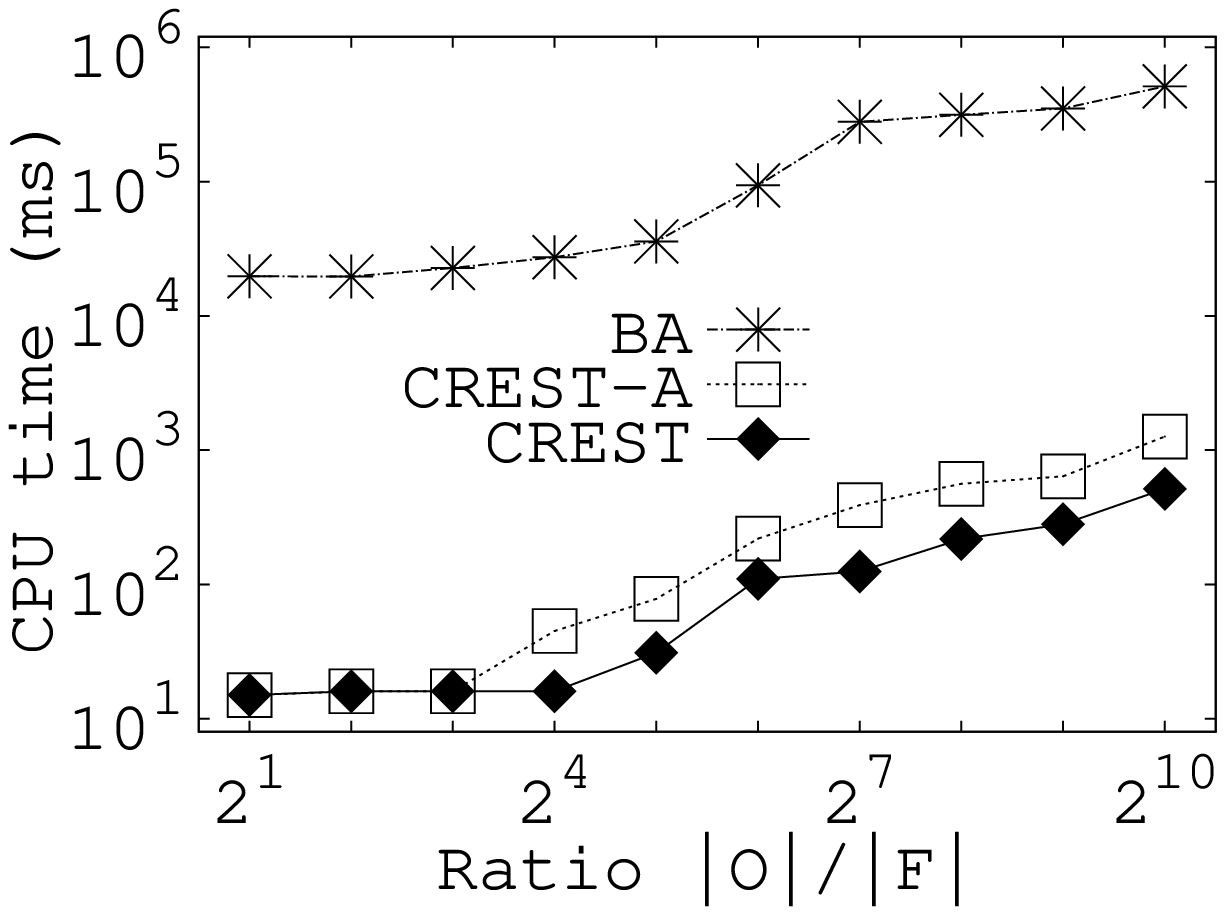, width=1.55in}
\label{fig:cfZF}
}
\shrink
\caption{Effect of $\frac{|\mathcal{O}|}{|\mathcal{F}|}$ with L{\small 1} diatance}\label{fig:cf}
\shrink 
\vspace{-5pt}
\end{figure}

\textbf{Effect of Data Set Size}.
We then fix $\frac{|\mathcal{O}|}{|\mathcal{F}|}$ at $2^7$ and vary the size of the client set $\mathcal{O}$ from $2^7$ to $2^{16}$. The results are plotted in Fig.~\ref{fig:n}. When the size of $\mathcal{O}$ is greater than $2^{13}$, the baseline runs for more than $24$ hours and
is early terminated, hence the results are not presented.
\begin{figure}[htb]
\centering
\subfigure[LA]{
\epsfig{file=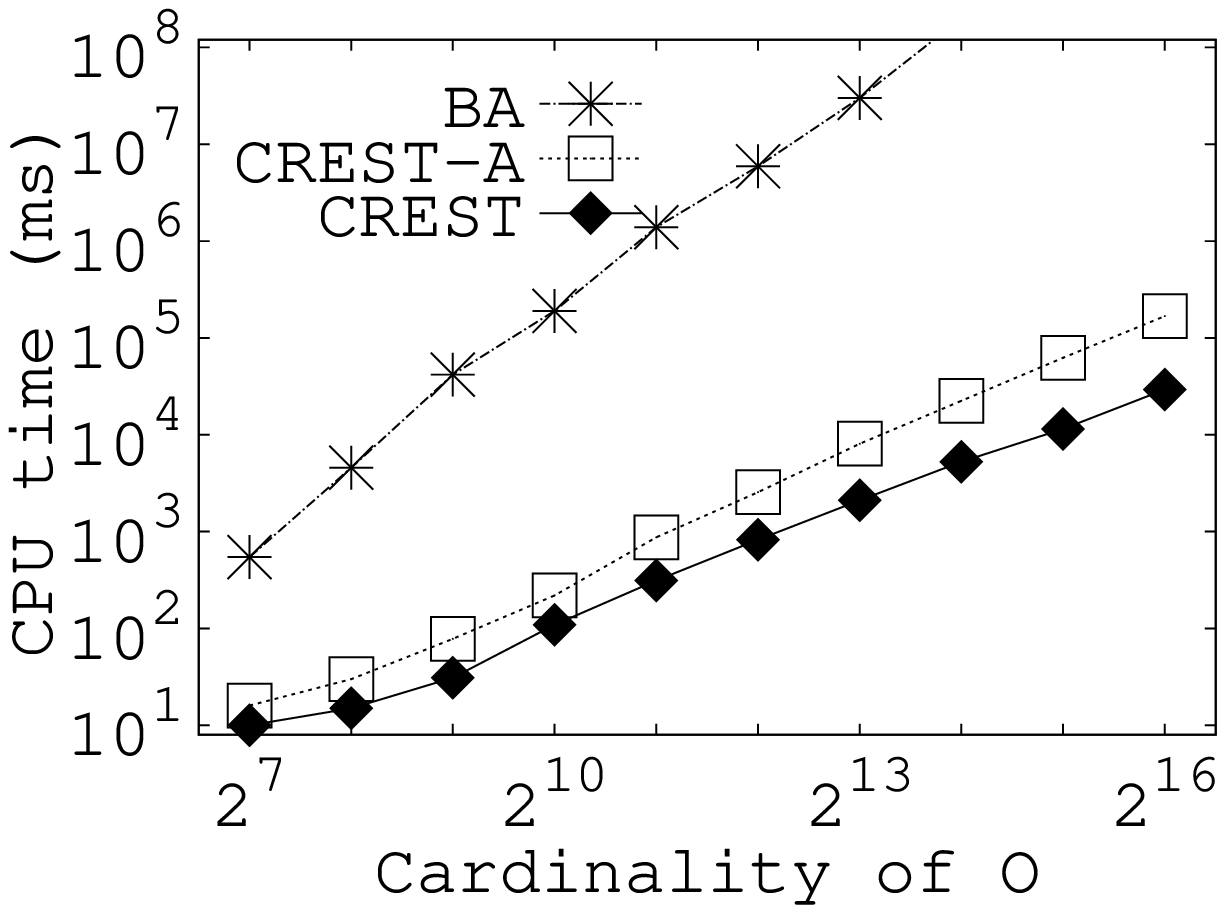, width=1.5in}
\label{fig:nLA}
}
\subfigure[NYC]{
\epsfig{file=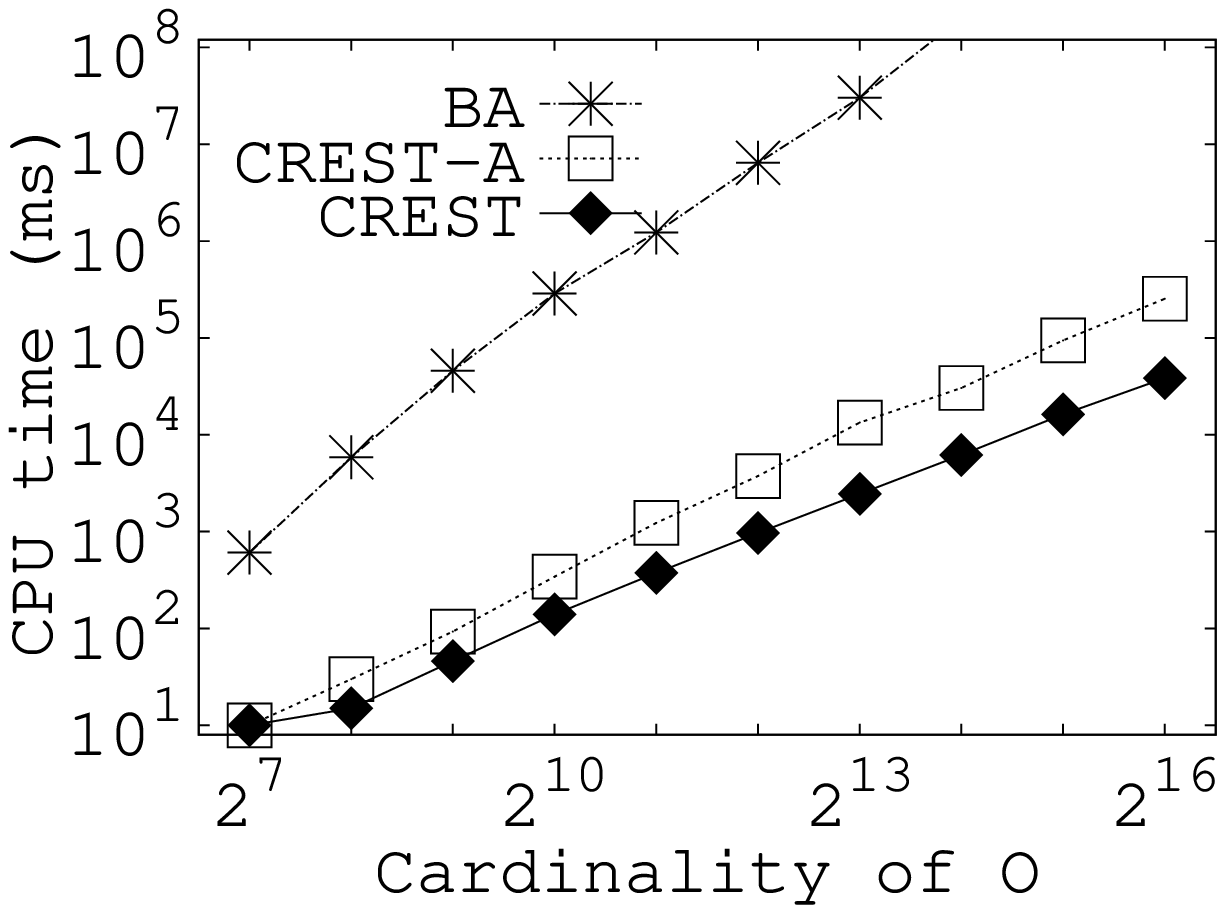, width=1.5in}
\label{fig:nNY}
}
\subfigure[Uniform]{
\epsfig{file=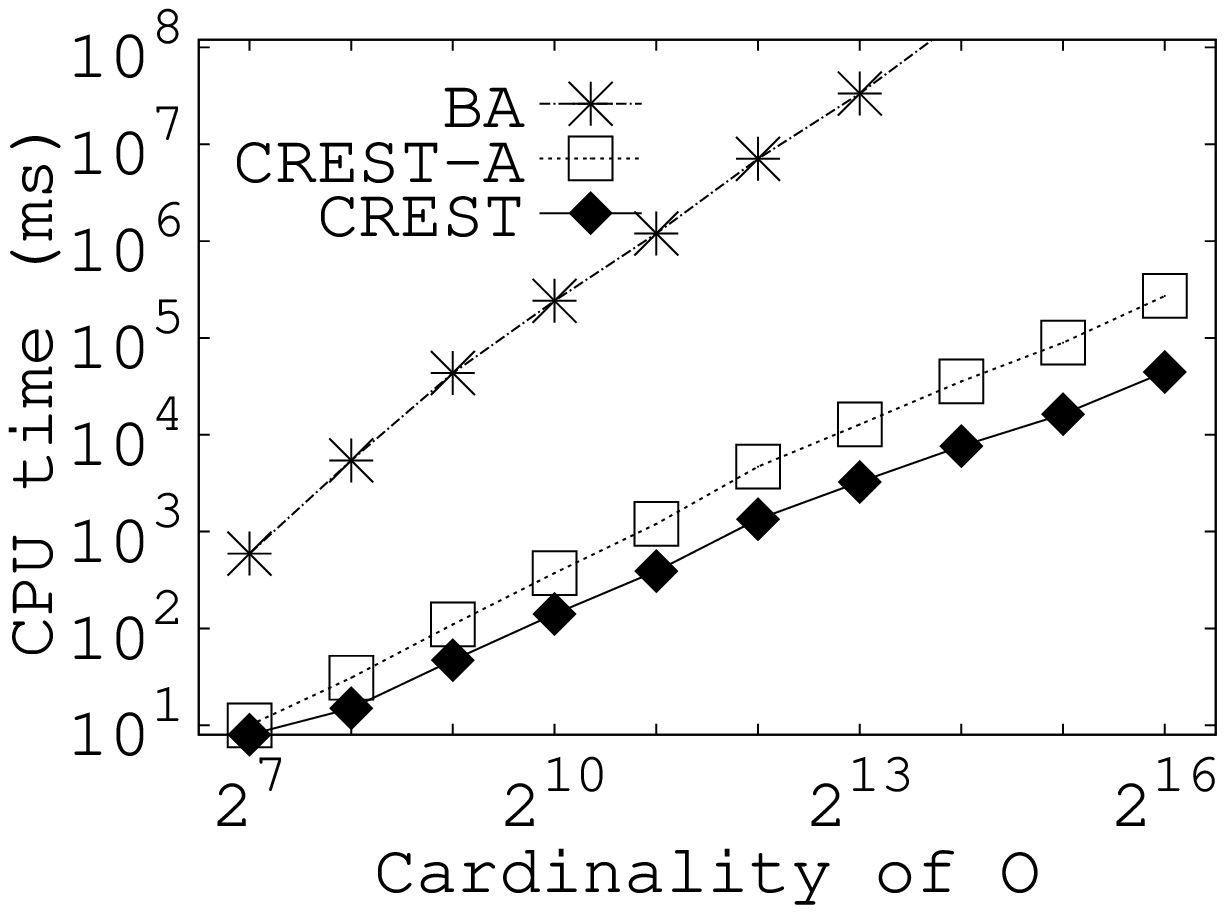, width=1.5in}
\label{fig:nUF}
}
\subfigure[Zipfian]{
\epsfig{file=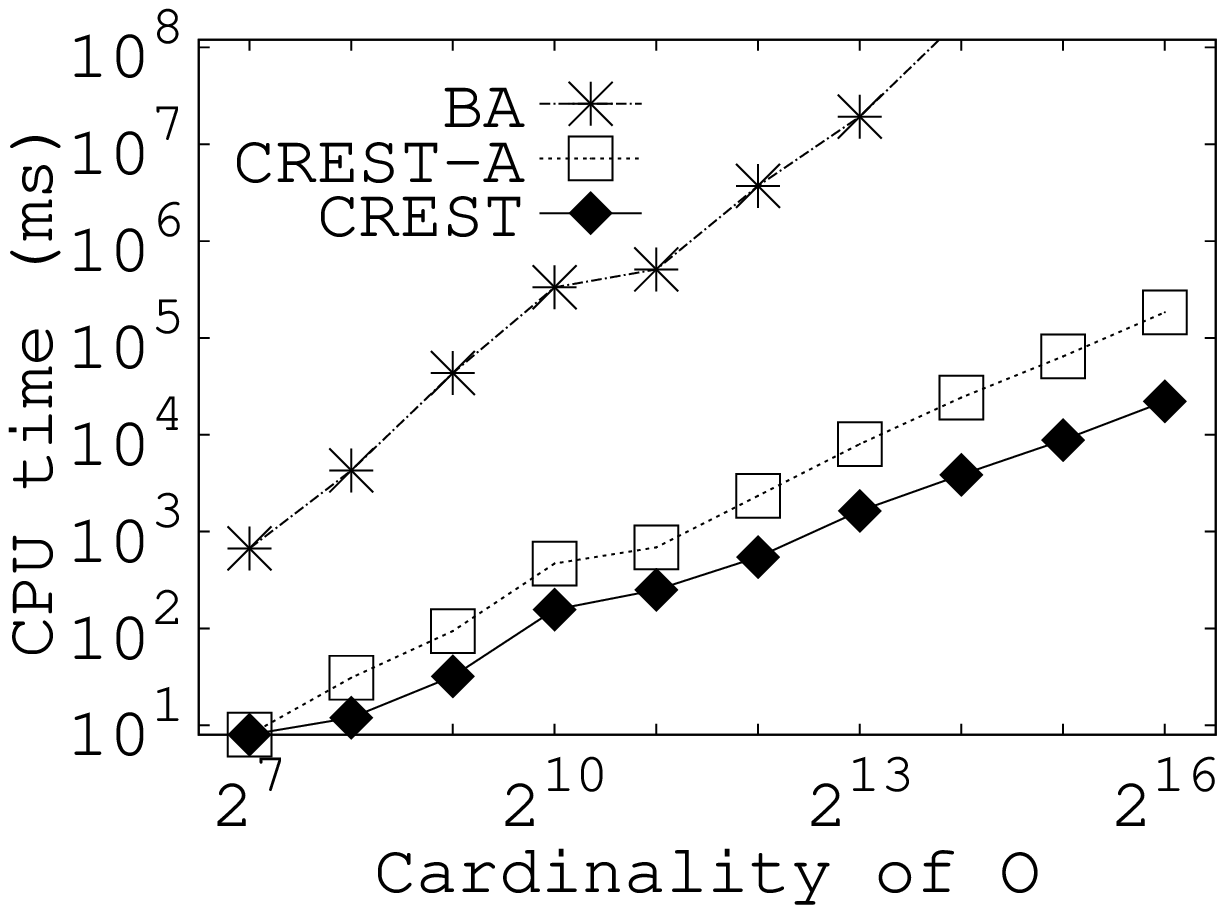, width=1.5in}
\label{fig:nZF}
}
\shrink
\caption{Effect of data set size with L{\small 1} diatance}\label{fig:n}
\shrink
\end{figure}

Again we can see that CREST outruns the baseline by at least three orders of magnitude and outruns CREST-A by up to an order of magnitude. The running time of the baseline increases much faster than that of the other two algorithms, which indicates the number of point enclosure queries computed in the baseline increases dramatically when $n = |\mathcal{O}|$ becomes larger.
The growth rate of CREST-A is also higher than that of CREST. This implies that the number of times of repeated labeling becomes larger with the increase of data size.
The lowest growth rate of CREST also demonstrates its scalability for processing much larger data sets.

%

\subsection{Performance of CREST with L{\small 2} Distance}
\vspace{-1pt} 
We repeat the above experiments with the $L_2$ distance metric, where the CREST algorithm for $L_2$ (\textbf{CREST-L2}) is compared with the pruning algorithm (\textbf{Pruning})
described in Section~\ref{subsec:L2}.
We compute the influence with the function in~\cite{Sun2012cikm} and use the two algorithms to find the regions with the \emph{maximum} influence, since in such settings the pruning algorithm performs the best.
This also shows the flexibility of CREST since the adaptation to various influence functions and supporting post-processing operations is very easy.

\begin{figure}[htb]
\centering
\subfigure[LA]{
\epsfig{file=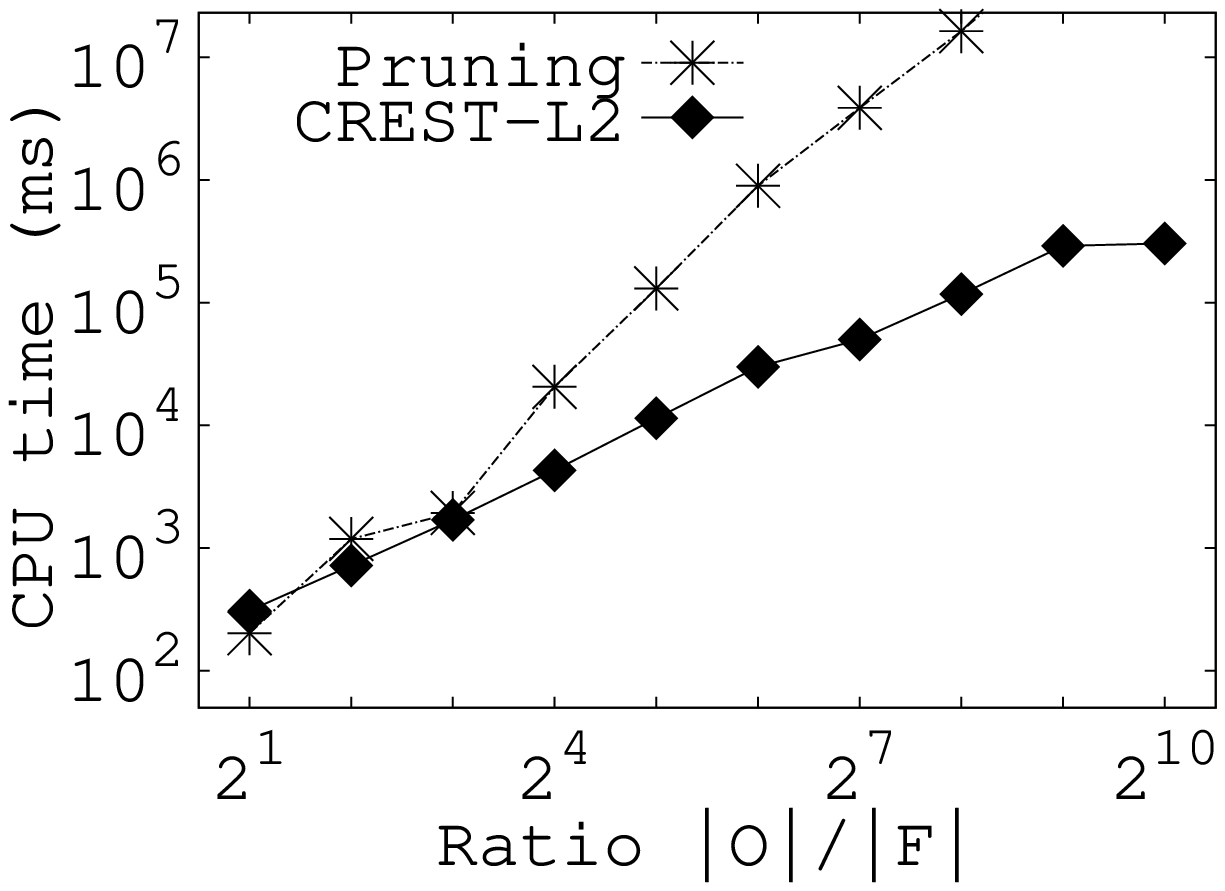, width=1.5in}
\label{fig:cfLA_L2}
}
\subfigure[NYC]{
\epsfig{file=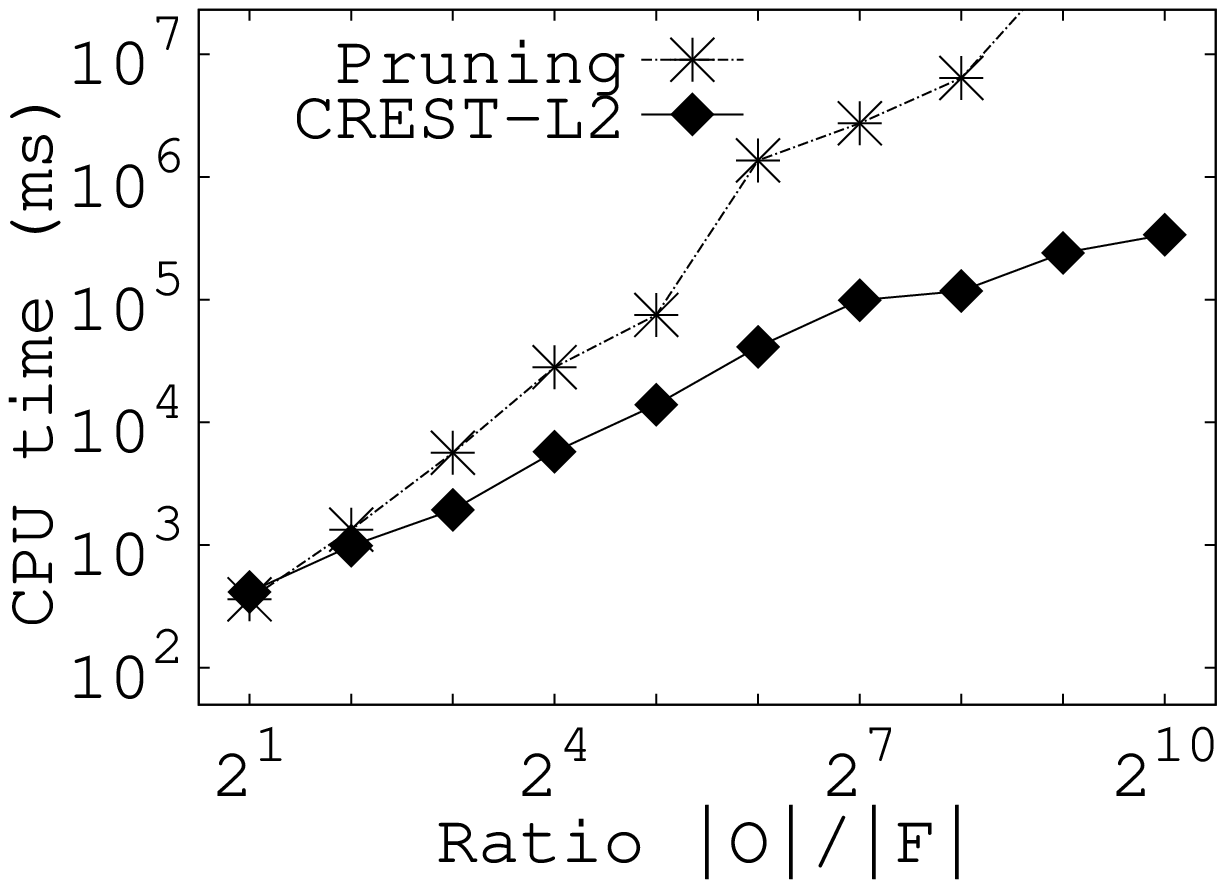, width=1.5in}
\label{fig:cfNY_L2}
}
\subfigure[Uniform]{
\epsfig{file=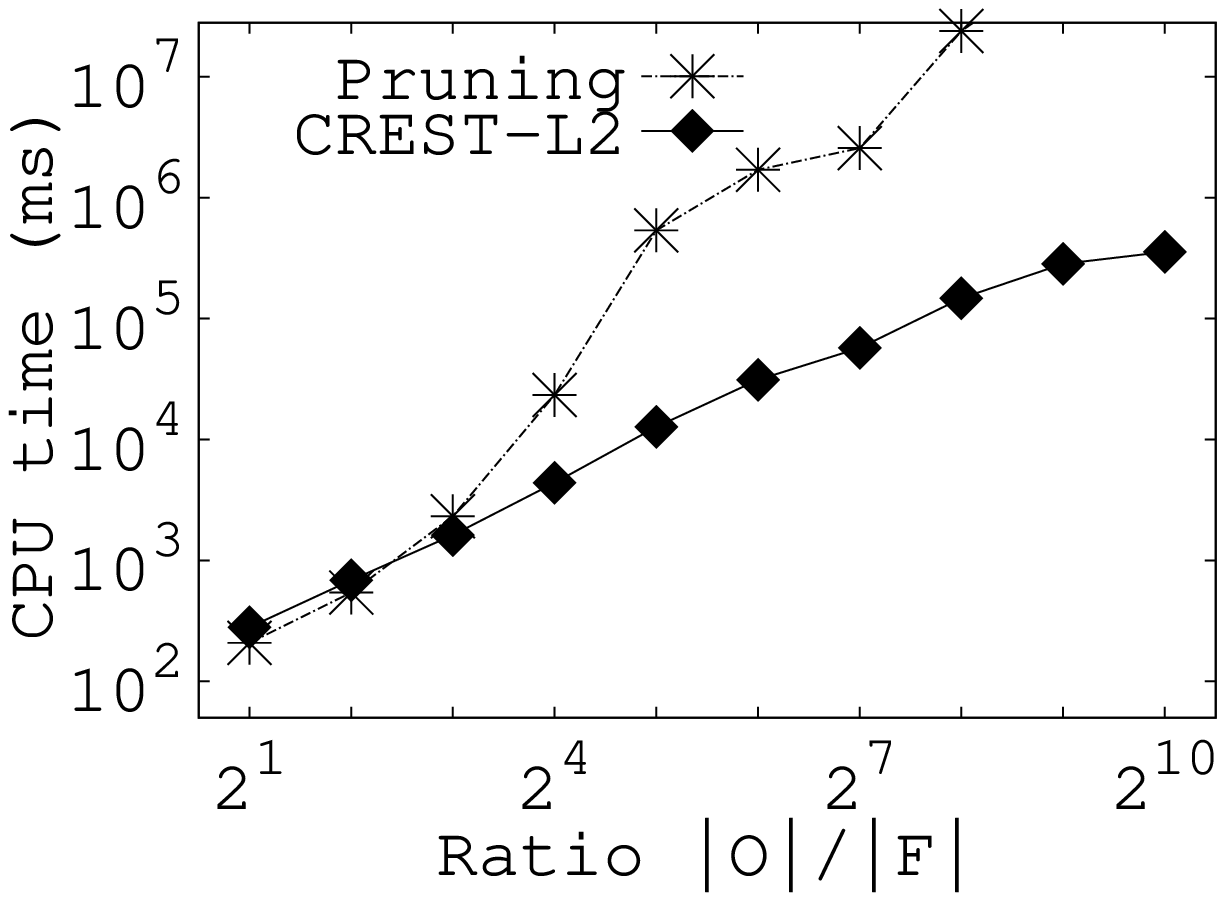, width=1.5in}
\label{fig:cfUF_L2}
}
\subfigure[Zipfian]{
\epsfig{file=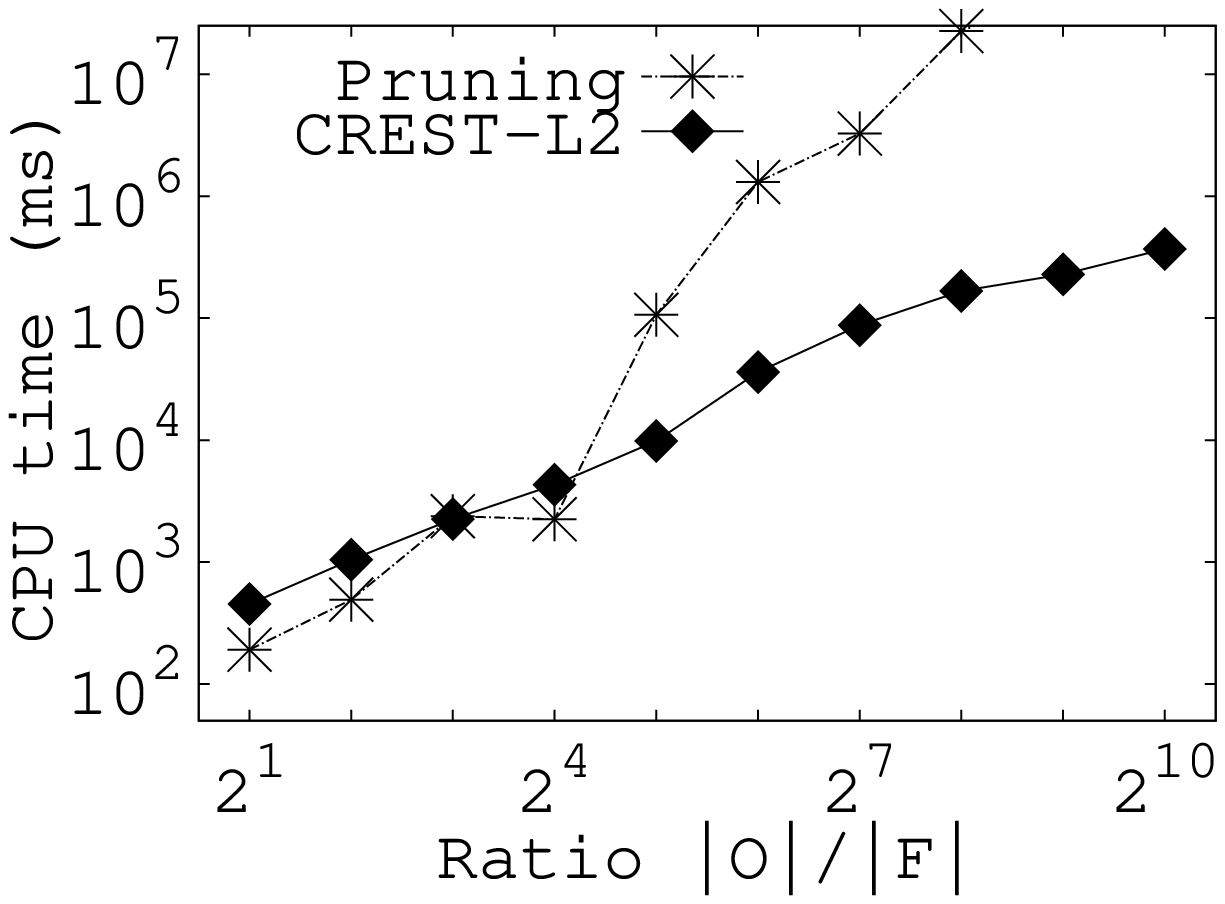, width=1.5in}
\label{fig:cfZF_L2}
}
\shrink
\caption{Effect of $\frac{|\mathcal{O}|}{|\mathcal{F}|}$ with L{\small 2} distance}\label{fig:L2_r}
\shrink
\end{figure}

\textbf{Effect of $\frac{|\mathcal{O}|}{|\mathcal{F}|}$}.
We first vary the ratio from $2^1$ to $2^{10}$ and plot the results in Fig.~\ref{fig:L2_r}. In all data sets, when the ratio is greater than $2^3$, we can see that CREST consistently outperforms the Pruning algorithm by several orders of magnitude. With the increase of $\frac{|\mathcal{O}|}{|\mathcal{F}|}$, the performance of the Pruning algorithm deteriorates rapidly, since the number of regions enumerated grows exponentially with the increase of $\frac{|\mathcal{O}|}{|\mathcal{F}|}$. Comparing with the Pruning algorithm, CREST has a much lower growth rate. When $\frac{|\mathcal{O}|}{|\mathcal{F}|}$ is less than or equal to $2^2$, in Fig.~\ref{fig:cfZF_L2} the Pruning algorithm runs slightly faster than CREST. This is because the number of regions enumerated in the Pruning algorithm is small, while the number of events in CREST is large when the data distribution is very skewed. Overall, CREST still outruns the Pruning algorithm by up to three orders of magnitude.

\textbf{Effect of Data Set Size}.
Next, we fix $\frac{|\mathcal{O}|}{|\mathcal{F}|}$ at $2^5$ and vary the size of $|\mathcal{O}|$ from $2^7$ to $2^{16}$. The results are presented in Fig.~\ref{fig:L2_n}. In all the data sets, again CREST consistently outperforms the Pruning algorithm. With the increase of $|\mathcal{O}|$, CREST and the Pruning algorithm have a similar growth rate. The running time of the Pruning algorithm gets closer to that of CREST when $|\mathcal{O}|$ is very large. This is because although the number of regions increases with the increase of $|\mathcal{O}|$, most of them are pruned without being searched in the Pruning algorithm, which does not happen in CREST. 
It is notable that even we solve the maximization problem with CREST which is quite general,
it still outruns the specialized Pruning algorithm designed for the problem, which demonstrates the efficiency of CREST.

\begin{figure}[htb]
\centering
\subfigure[LA]{
\epsfig{file=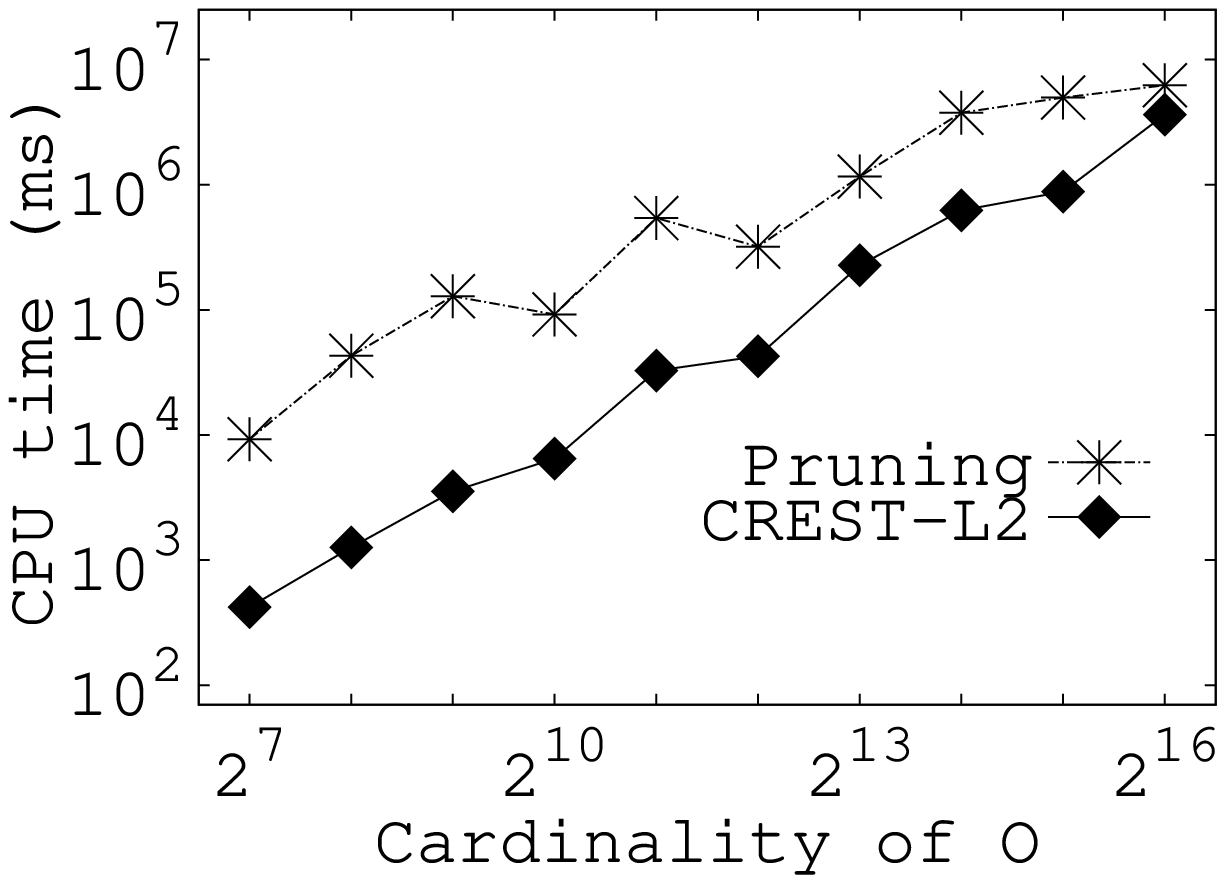, width=1.5in}
\label{fig:nLA_L2}
}
\subfigure[NYC]{
\epsfig{file=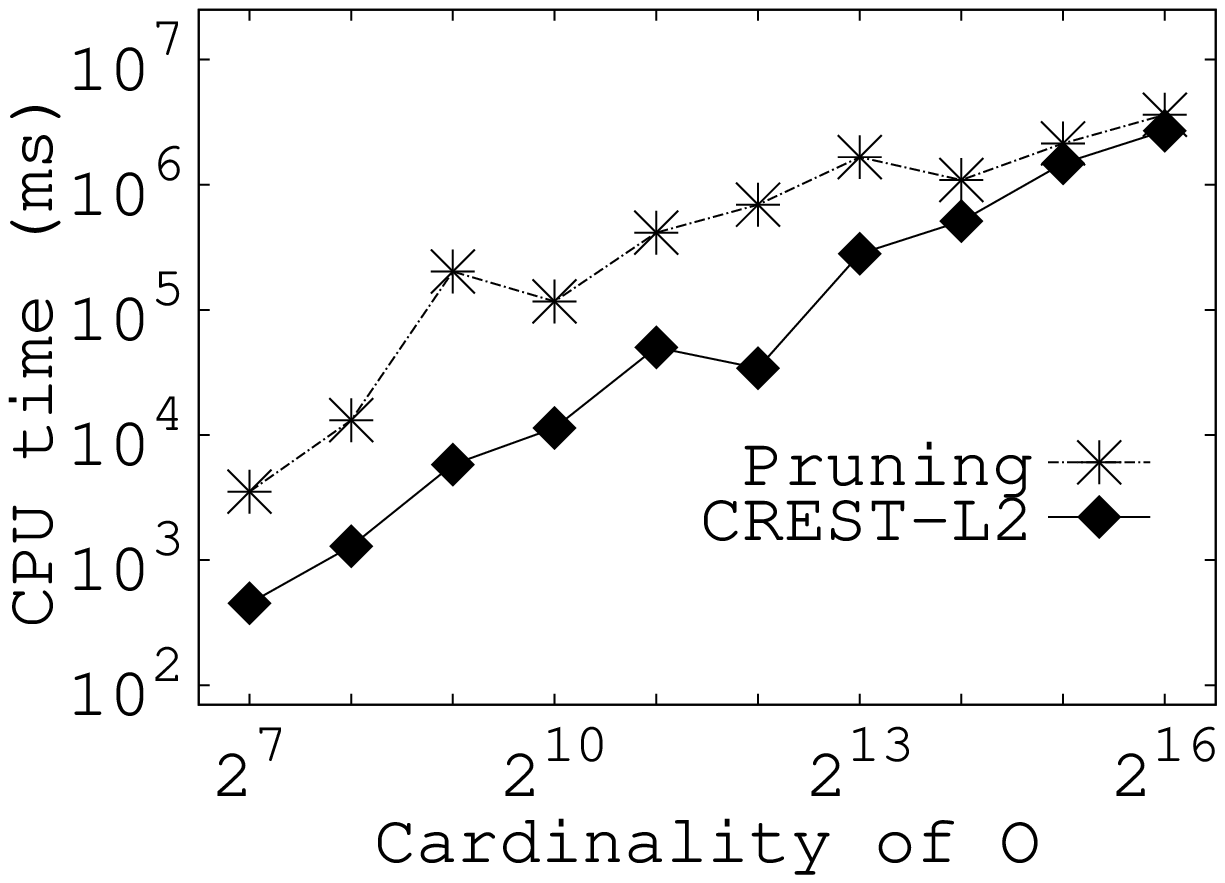, width=1.5in}
\label{fig:nNY_L2}
}
\subfigure[Uniform]{
\epsfig{file=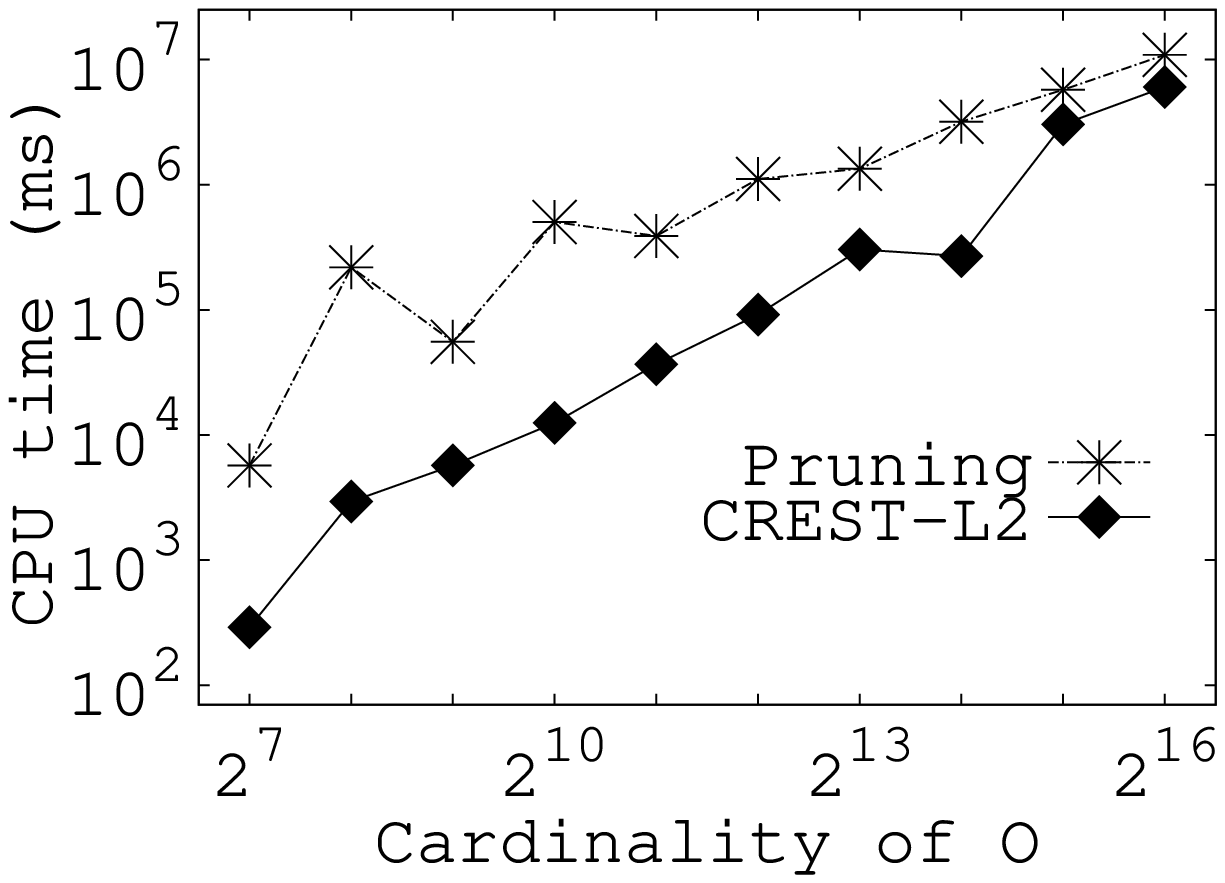, width=1.5in}
\label{fig:nUF_L2}
}
\subfigure[Zipfian]{
\epsfig{file=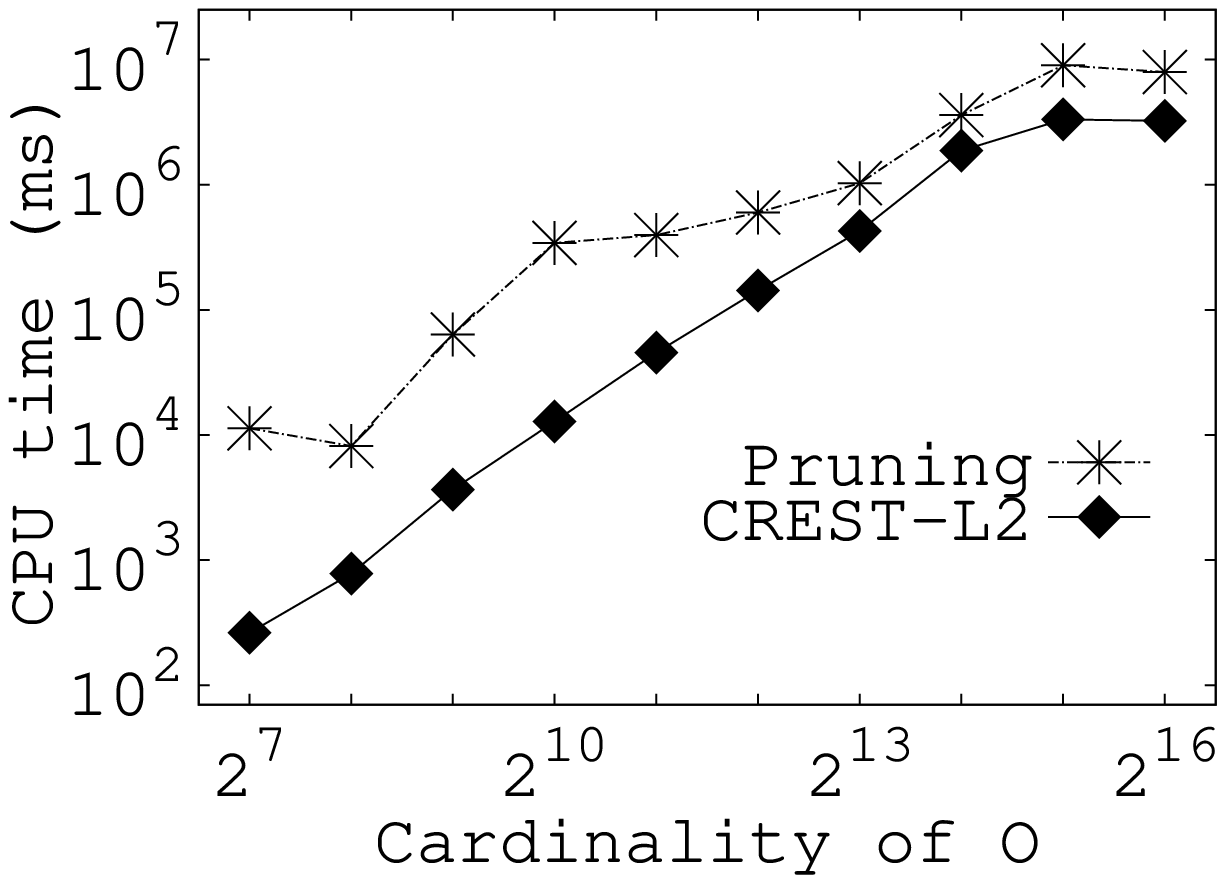, width=1.5in}
\label{fig:nZF_L2}
}
\shrink
\caption{Effect of data set size with L{\small 2} distance}\label{fig:L2_n}
\shrink
\end{figure}

\vspace{5pt}
\section{Conclusions} \label{sec:conclu}

In this paper, we proposed the RNN heat map problem, which computes the influence of every point in the space. 
Comparing to existing studies which give only the points or regions with the highest influence, the RNN heat map enables exploring the influence of the whole space while considering qualitative factors at any instant during the exploration.
We solved the problem by first reducing it to the Region Coloring (RC) problem, and then computing the influence on regions instead of points with a novel algorithm called CREST. We proposed two techniques in CREST, one to avoid point enclosure queries in the influence computation and the other to reduce the total number of times of the influence computation.
Through a detailed analysis, we showed that the number of influence computation in CREST is asymptotically optimal. We also showed that the worst-case time complexity of CREST is much lower than that of the baseline algorithm and in many cases meets the lower bound of RC. We conducted extensive experiments on both real and synthetic data sets. The results showed that CREST outperforms alternative algorithms by up to three orders of magnitude.

{\small
\textbf{Acknowledgments.}
This work is partially supported by the National Natural Science Foundation of China (Nos. 61402155 and 61432006). Rui Zhang is supported by ARC Future Fellow project FT120100832. Jianzhong Qi is supported by Melbourne School of Engineering Early Career Researcher Grant (No. 4180-E55) and University of Melbourne Early Career Researcher Grant (No. 603049).
\vspace{-10pt}
}

{
\bibliographystyle{abbrv}
{\small
\bibliography{ref}
}
}
\end{document}